\documentclass[a4paper,11pt]{article}
\pdfoutput=1 

\usepackage{jheppub} 

\usepackage[T1]{fontenc} 
\usepackage{subcaption}
\usepackage{verbatim}
\usepackage{mathrsfs}
\usepackage{physics}
\usepackage{amsthm}
\usepackage{mathtools}
\usepackage{hyperref}
\usepackage{enumerate}

\newcommand{\Del}{\nabla}
\newcommand{\del}{\partial}

\renewcommand{\tilde}{\widetilde}
\renewcommand{\bar}{\overline}

\newcommand{\reals}{\mathbb{R}}
\newcommand{\comps}{\mathbb{C}}

\newcommand{\A}{\mathcal{A}}
\newcommand{\M}{\mathcal{M}}
\renewcommand{\H}{\mathcal{H}}

\renewcommand{\hat}{\widehat}
\newcommand{\supp}{\operatorname{supp}}
\newtheorem{theorem}{Theorem}

\newtheorem{definition}[theorem]{Definition}

\theoremstyle{definition}

\numberwithin{theorem}{section}

\title{Analyticity and the Unruh effect: a study of local modular flow}
\author{Jonathan Sorce}
\affiliation{Center for Theoretical Physics, Massachusetts Institute of Technology, 182 Memorial Drive, Cambridge, MA, USA}
\abstract{
The Unruh effect can be formulated as the statement that the Minkowski vacuum in a Rindler wedge has a boost as its modular flow.
In recent years, other examples of states with geometrically local modular flow have played important roles in understanding energy and entropy in quantum field theory and quantum gravity.
Here I initiate a general study of the settings in which geometric modular flow can arise, showing (i) that any geometric modular flow must be a conformal symmetry of the background spacetime, and (ii) that in a well behaved class of ``weakly analytic'' states, geometric modular flow must be future-directed.
I further argue that if a geometric transformation is conformal but not isometric, then it can only be realized as modular flow in a conformal field theory.
Finally, I discuss a few settings in which converse results can be shown --- i.e., settings in which a state can be constructed whose modular flow reproduces a given vector field.
}


\begin{document}
\maketitle

\section{Introduction}

In 1975, Hawking made the remarkable discovery that semiclassical black holes evaporate by emitting thermal radiation \cite{hawking1975particle}.
In closely related work, it was soon discovered that for quantum fields propagating in the background of an eternal Schwarzschild black hole, there is a natural ``vacuum state'' that appears thermal with respect to Schwarzschild time \cite{hartle1976path}.
Around the same time, Unruh showed that a similar phenomenon --- the ``Unruh effect'' --- occurs even in Minkowski spacetime, where observers following the trajectory of a boost will experience the global vacuum as a thermal bath \cite{unruh1976notes, unruh1984acceleration}.
It was eventually realized that a version of the Unruh effect is present in a large class of spacetimes with Killing symmetry \cite{Sewell:1982zz, Kay:1988mu, jacobson1994note}.
In all of these settings, the surprise is the emergence of thermodynamic behavior in non-equilibrium states.
These states are globally pure, but they appear thermal to a preferred class of observers.

A precise way of formulating what it means for a state to have ``emergent thermodynamics'' is provided by Tomita-Takesaki modular theory.
As explained in \cite{Sorce:modular} (see also the textbook \cite{Takesaki:volII}), for \textit{any} state $|\Psi\rangle$ in quantum field theory, and any causally complete spacetime region $A$, if $|\Psi\rangle$ is sufficiently entangled between $A$ and its spatial complement, then there is a \textit{unique} self-adjoint operator $K_{\Psi},$ called the modular Hamiltonian, such that the subsystem $A$ looks thermal in the state $|\Psi\rangle$ with respect to the time-evolution map $e^{- i K_{\Psi} t}.$\footnote{The statement that a subsystem ``looks thermal'' in a state is made precise using a technical tool called the KMS condition; see \cite[section 3]{Sorce:modular}, or section \ref{subsec:modular-flow} of the present paper.}
The Unruh effect was first formulated in this language by Bisognano and Wichmann \cite{bisognano1976duality}, who showed that in the framework of axiomatic quantum field theory, the modular Hamiltonian of the vacuum state restricted to a Rindler wedge is proportional to the generator of boosts in that wedge.
This is illustrated in figure \ref{fig:unruh-effect}.

\begin{figure}
	\centering
	\includegraphics[scale=1.2]{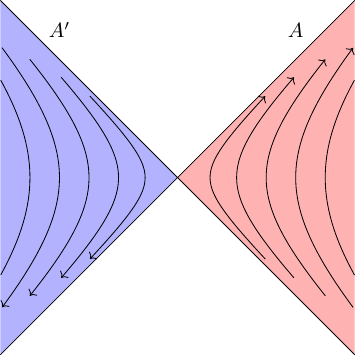}
	\caption{The Bisognano-Wichmann formulation of the Unruh effect.
	The modular flow of the Minkowski vacuum $|\Omega\rangle,$ restricted to a Rindler wedge, is the Lorentz boost generated by $\xi^a = 2 \pi \left[x \left( \frac{\partial}{\partial t} \right)^a + t \left(\frac{\partial}{\partial x} \right)^a\right]$.}
	\label{fig:unruh-effect}
\end{figure}

In this language, we say that a state in quantum field theory exhibits emergent non-equilibrium thermodynamics if there is a subregion of spacetime in which its modular flow is geometrically local.
Every sufficiently entangled state has \textit{some} Hamiltonian with respect to which it looks thermal; the cases of greatest interest are the ones where that Hamiltonian describes a physical family of observers.

There are a few known examples of states in various quantum field theories, in various spacetimes, where modular flow is geometrically local.
Beyond the settings described above, it is also known that modular flow is geometrically local in (i) the vacuum of a conformal field theory (CFT) restricted to a ball in a conformally flat spacetime \cite{Hislop:CFT-ball, Casini:2011kv, Frob:2023hwf}; (ii) a global thermal state in two-dimensional CFT restricted to a Rindler wedge in Minkowski spacetime \cite{Borchers:2D}; and (iii) a global thermal state in two-dimensional CFT restricted to an interval in Minkowski spacetime \cite{Wong:2013gua, Cardy:modhams}.

In addition to their conceptual interest, states with geometric modular flow have found many remarkable applications.
They have been used to study the universal entanglement structure of subregions in quantum field theory \cite{Hislop:CFT-ball, Fredenhagen:1984dc, Hollands:edge}, to compute entanglement quantities explicitly in CFT \cite{Cardy:modhams}, to prove entropic versions of renormalization monotonicity theorems \cite{Casini:c-theorem, Casini:F-theorem, Casini:a-theorem}, and even to constrain the complexity hierarchy of bulk reconstruction in holography \cite{Engelhardt:pythons, Chen:2022eyi}.
Most recently, states with geometric modular flow have been used to argue that a perturbative treatment of diffeomorphism invariance in quantum gravity leads to finite renormalized entropy differences for static anti-de Sitter black holes \cite{Witten:crossed-product, Chandrasekaran:large-N, Penington:JT}, the static patch of de Sitter spacetime \cite{Chandrasekaran:dS}, and more general spacetime regions \cite{Jensen:2023yxy, Kudler-Flam:2023qfl}.\footnote{While they will not be discussed in this paper, there are also states with modular flow that is instantaneously local on a slice \cite{Wall:second-law, Faulkner:ANEC, Casini:2011kv} or globally bilocal \cite{Casini:chiral, Longo:2009mn, Hollands:chiral,  Blanco:2019xwi, Fries:2019ozf, Mintchev:2020uom, Mintchev:2020jhc}.
Both of these kinds of states have found interesting applications; see e.g. \cite{Ceyhan:QNEC, Chen:2019iro, Leutheusser:HSMT, Jensen:2023yxy, Kudler-Flam:2023qfl}.}

Given the broad usefulness of states with geometrically local modular flow, it is of considerable interest to understand these states in a general sense without needing to treat them on a case-by-case basis.
The purpose of this paper is to initiate this general study by placing precise constraints on the settings in which local modular flow can arise.
I emphasize that the main results of this paper are \textit{constraints} --- statements of the form ``whenever local modular flow exists, it must have property X'' --- rather than constructive statements about the existence of local flow in certain settings.
A complete characterization of geometric modular flow would be an interesting subject for future work, and the techniques introduced here --- in particular, the analyticity techniques used in section \ref{sec:future-directed} --- seem well suited for future investigations.

The plan of the paper follows.

\begin{itemize}
	\item In \hyperref[sec:assumptions]{\textbf{section 1.1}}, I list my choices of convention along with some assumptions that are made to simplify the presentation of the rest of the paper.
	In particular, I explain that all formulas will be expressed for theories generated by a single scalar field, but indicate how to generalize to arbitrary theories.
	\item In \hyperref[sec:toolkit]{\textbf{section 2}}, I give a brief, pedagogical introduction to the mathematical tools that are needed in the rest of the paper.
	I give the basic definitions and properties of modular flow, explain the connection between analyticity and positive energy in quantum field theory, and describe a powerful technique --- the ``disk method'' --- for analytically continuing functions of multiple complex variables.
	\item In \hyperref[sec:constraints]{\textbf{section 3}}, I determine a set of necessary conditions for a geometric flow in spacetime to be the modular flow of a state.
	\begin{itemize}
		\item In \hyperref[sec:conformality]{\textbf{section 3.1}}, I show that due to microcausality, any unitary group which acts geometrically must implement a conformal transformation, implying that geometric modular flow must be generated by a conformal Killing vector field.\footnote{This constraint is known to some practitioners of modular theory --- for example, a special case appears in \cite{borchers2000revolutionizing}, and a few people have indicated to me via private communication that they are aware of the general constraint --- but to my knowledge the details have not been worked out in the literature for general theories and general spacetimes.
		See however \cite[section 2]{Chen:2022eyi} where a similar claim is made for states that are prepared by Euclidean path integrals.}
		\item In \hyperref[sec:future-directed]{\textbf{section 3.2}}, I rule out non-future-directed conformal transformations as candidates for modular flow in a class of states that includes the ``analytic states'' discussed by Strohmaier and Witten in the context of the timelike tube theorem \cite{Strohmaier:analytic-states, Strohmaier:timelike-tube}.
		I also present a different way of thinking about analytic states than the one given in \cite{Strohmaier:analytic-states, Strohmaier:timelike-tube}, skimming over the thorny details of microlocal analysis to explain an analytic state as one for which nonanalytic singularities can be resolved by an $i\epsilon$ prescription with the same sign as the one used for vacuum correlators in Minkowski spacetime.
		The proof that geometric modular flow is future-directed in my well behaved class of states is accomplished by microlocalizing a result due to Trebels for the Minkowski vacuum \cite{Trebels}.
		\item In \hyperref[sec:isometric]{\textbf{section 3.3}}, I argue against the possibility of having a non-isometric conformal transformation as the modular flow of a state in a non-conformal field theory.
		Unlike the arguments of sections \ref{sec:conformality} and \ref{sec:future-directed}, the argument of section \ref{sec:isometric} is not rigorous, and is included as a suggestive direction for future work.
	\end{itemize}
	\item In \hyperref[sec:sufficiency-construction]{\textbf{section 4}}, I explain some settings in which path integrals can be used to construct a state whose modular flow matches a given geometric transformation of spacetime.
	\item In \hyperref[sec:discussion]{\textbf{section 5}}, I comment on three unsolved problems that I think will be important in the future: (i) the application of analytic techniques to non-analytic spacetimes, (ii) the scaling of local modular flow near the edge of a region, and (iii) the construction of states with local modular flow in the absence of a ``Euclidean section.''
\end{itemize}

\noindent 
Appendix \ref{app:axis-extension} explains a family of contours in higher-dimensional complex spaces that are used to apply the disk method (section \ref{subsec:disk-method}) to constraining modular flow in weakly analytic states (section \ref{sec:future-directed}).

\vspace{0.5cm}
\hrule
\hrule
\vspace{0.5cm}

\noindent
\textbf{Note added (April 23, 2025):} About one year after this paper first appeared online, Rainer Verch informed me of overlapping work in the mathematical physics literature that reaches similar conclusions using different assumptions or techniques.
\begin{itemize}
	\item In \cite{Keyl:1993ye}, Keyl studied geometrically local ``net automorphisms'' of C$^*$ algebras and concluded that any such automorphism must be conformally implemented.
	This is similar to the conclusion of section \ref{sec:conformality}, but Keyl's proof assumes that field operators commute \textit{only} at spacelike separation.
	This is a stronger assumption than the one made in section \ref{sec:assumptions}, and consequently Keyl's proof does not apply, for example, to flat-spacetime conformal field theories in two dimensions.
	\item 
	In \cite{Pinamonti:2018ltu}, Pinamonti, Sanders, and Verch studied connections between the smooth wave front set and the existence of spacelike-directed geometric flows satisfying the KMS condition.
	Their technique is similar to the one outlined in section \ref{sec:conclusion-analyticity}. 
	In some sense their assumptions should apply in a larger class of spacetimes than our assumptions, because they use only the smooth wave front set, instead of the analytic wave front set that we use in section \ref{sec:analytic-states}.
	On the other hand, while we assume that the analytic wave front set must \textit{not} contain certain analytic directions, the authors of \cite{Pinamonti:2018ltu} assume that the smooth wave front set \textit{must} contain certain singular directions.
	The relationship between the smooth wave front set $\text{WF}_{\text{S}}$ and the analytic wave front set $\text{WF}_{\text{A}}$ is
	\begin{equation} \label{eq:wavefront-inclusion}
		\text{WF}_{\text{S}} \subseteq \text{WF}_{\text{A}}.
	\end{equation}
	Roughly speaking, there is a set $K$ such that we assume $\text{WF}_{\text{A}} \subseteq K,$ and such that \cite{Pinamonti:2018ltu} assumes $K \subseteq \text{WF}_{\text{S}}.$
	Because of the direction of the inclusion in \eqref{eq:wavefront-inclusion}, these assumptions are incomparable.
	It would be very interesting to undertake a closer analysis of the connections between these two works.
\end{itemize}
\hrule
\hrule

\subsection{Assumptions and conventions}
\label{sec:assumptions}

In my first attempt to write this paper, I stated everything in the greatest generality possible.
This ended up causing a lot of bloat in the explanations without providing any real insight.
To make the paper as clean and informative as possible, I have therefore elected to make a few universal assumptions:
\begin{itemize}	
	\item Every field theory in this paper is relativistic.
	\item I will assume that every spacetime is globally hyperbolic, i.e., that it admits a complete Cauchy slice.
	It should be understood, however, that all results generalize to spacetimes that are not globally hyperbolic but that are still causally well behaved, like anti-de Sitter spacetime.
	\item I will only discuss field theories that are generated by a single real scalar field $\phi(x).$
	There is no problem generalizing to other kinds of field theories.
	All results of this paper apply to field theories satisfying the assumption
	\begin{quote}
		\textit{for any pair of null-separated points $x$ and $y$, there exists a pair of observable fields $\phi$ and $\psi$ with $[\phi(x), \psi(y)] \neq 0.$}
	\end{quote}
	The advantage of restricting our attention to theories generated by a single real scalar field is that we can replace this with the simpler assumption
	\begin{quote}
		\textit{for any pair of null-separated points $x$ and $y$, we have $[\phi(x), \phi(y)] \neq 0.$}
	\end{quote}
	However, all arguments in the paper generalize in a straightforward way to field theories satisfying the more general assumption.
\end{itemize}

\noindent
In addition to the above assumptions, there are several notational conventions that must be settled, and my choices are as follows:
\begin{itemize}
	\item The metric signature is ``mostly pluses,'' so that timelike vectors have negative norm-squared.
	\item The dimension of spacetime is $d+1$.
	\item In Minkowski spacetime, the global momentum operator $P_{\mu}$ is defined so that the operator $e^{-i P_{\mu} \xi^{\mu}}$ translates states by the vector $\xi^{\mu}.$
	With this convention, if $x^0$ is a time coordinate, then the Hamiltonian that evolves in the direction $x^0$ is the lowered-index operator $P_0,$ \textit{not} the raised-index operator $P^0,$ which differs from $P_0$ by a sign.
\end{itemize}

\noindent
Finally, I will frequently refer to objects like $\langle \phi(x_1) \phi(x_2) \rangle$ or $\phi(x_1) \phi(x_2) | \Psi \rangle$ as ``functions'' of $x_1$ and $x_2,$ even though they are only defined as distributions.

\section{Setting up the toolkit}
\label{sec:toolkit}

This section describes the basic mathematical tools needed in the rest of the paper.
Section \ref{subsec:modular-flow} explains the algebraic approach to studying quantum fields, and the basics of modular flow.
Section \ref{subsec:field-analyticity} explains the relationship between positivity of energy and analyticity of position-space field correlators in Minkowski spacetime.
Section \ref{subsec:disk-method} explains a tool in multivariable complex analysis that can be used to analytically continue functions defined in a small domain of $\comps^n$ to a much larger one; it can be skipped by readers willing to accept the results of section \ref{sec:future-directed} without understanding the details of the proof.

\subsection{Modular flow in quantum field theory}
\label{subsec:modular-flow}

Let $\M$ be a globally hyperbolic spacetime, and let $\H$ be a Hilbert space describing a sector of a quantum field theory on $\M.$
As explained in section \ref{sec:assumptions}, I will write all expressions as though the field theory is generated by a single scalar field $\phi(x)$, though this does not affect the generality of the stated results.

Given any smooth, compactly supported function $f,$ there is an associated operator
\begin{equation}
	\phi[f]
	= \int d^{d+1} x \sqrt{|g|} \phi(x) f(x).
\end{equation}
This is sometimes called a \textit{smeared field operator} or just a \textit{field operator}, with \textit{smearing function} $f.$
Field operators are unbounded; it is usually easier to work with bounded functions of the field operators, for example $e^{i\phi[f]}$ when $\phi[f]$ is Hermitian.
Corresponding to any spacetime region, there is an algebra of bounded operators consisting of polynomials of bounded functions of the smeared fields in that region.
Taking the completion of this algebra in an appropriate topology gives a von Neumann algebra.\footnote{A review of von Neumann algebras for physicists can be found in \cite{Sorce:types}.}

Now consider a region $A$ which is the domain of dependence of a partial Cauchy slice, and let $A'$ be its causal complement.
This is sketched in figure \ref{fig:partial-Cauchy}.
Let $\A$ be the von Neumann algebra generated by fields in $A$.
We will assume Haag duality, which is the statement that the von Neumann algebra $\A'$ generated by fields in $A'$ contains all operators that commute with $\A$.
In this setting, when $f$ is a smearing function supported in $A$, the field operator $\phi[f]$ is said to be \textit{affiliated} with $\A$ --- it is not strictly ``in'' $\A$ because it is unbounded, but it commutes with all operators in $\A'.$

\begin{figure}
	\centering
	\includegraphics[scale=1.25]{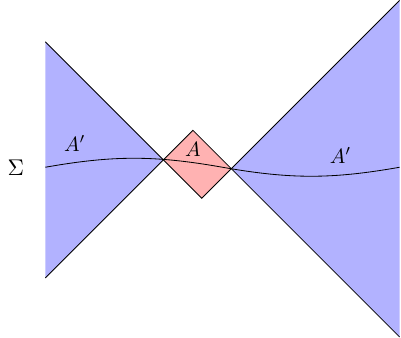}
	\caption{A Cauchy slice $\Sigma$ together with the domain of dependence $A$ for a partial Cauchy slice, and the complementary domain $A'.$}
	\label{fig:partial-Cauchy}
\end{figure}

Let $|\Psi\rangle$ be a state that is \textit{cyclic and separating} for $\A$, meaning that the set $\{a |\Psi\rangle\}_{a \in \A}$ is dense in $\H,$ and that $|\Psi\rangle$ is not annihilated by any nonzero operator in $\A$.
This is a notion that is well defined in the continuum, and which generalizes the concept of a lattice state having full rank across a bipartition; a state that is cyclic and separating for $\A$ may be thought of as being ``fully entangled'' between $A$ and $A'$.
Given a cyclic and separating state $|\Psi\rangle$, one can show that there exists a \textit{unique} self-adjoint operator $K_{\Psi}$, called the modular Hamiltonian of $|\Psi\rangle$ with respect to $\A$, satisfying the following three properties:
\begin{enumerate}[(i)]
	\item Modular flow preserves the state: $e^{- i K_{\Psi} t} |\Psi \rangle = |\Psi \rangle.$
	\item Modular flow preserves the algebras: For $a \in \A$ and $a' \in \A',$ we have
	\begin{equation}
		e^{i K_{\Psi} t} a e^{-i K_{\Psi} t} \in \A, \qquad e^{i K_{\Psi} t} a' e^{-i K_{\Psi} t} \in \A'.
	\end{equation}
	\item Modular flow satisfies the KMS condition: for $a, b \in \A,$ the time-evolved correlator $\langle \Psi | e^{i K_{\Psi} t} a e^{-i K_{\Psi} t} b |\Psi\rangle$ admits an analytic continuation $it \mapsto z$ for $z$ in the strip $0 \leq \text{Re}(z) \leq 1.$
	The analytic continuation is given by a function $F(z)$ satisfying
	\begin{align}
		F(it)
		& = \langle \Psi | e^{i K_{\Psi} t} a e^{-i K_{\Psi} t} b |\Psi\rangle, \\
		F(1+it)
		& = \langle \Psi | b e^{i K_{\Psi} t} a e^{-i K_{\Psi} t} |\Psi\rangle.
	\end{align}
\end{enumerate}
The existence and uniqueness of $K_{\Psi}$ is the fundamental theorem of Tomita-Takesaki theory, explained in detail in the textbook \cite{Takesaki:volII}.
A recent reformulation of the theory for physicists is given in \cite{Sorce:modular}.
As explained in \cite{Sorce:modular}, the KMS condition (condition (iii) above) is satisfied by thermal states with respect to a local Hamiltonian on a lattice, and is the appropriate continuum generalization of what it means for a state to ``look thermal'' with respect to a choice of time-evolution map.

For the purposes of this paper, the most important feature of modular flow is its domain of analytic continuation.
Given an operator $a \in \A$, the state $a |\Psi\rangle$ is in the domain of every unbounded operator $e^{- z K_{\Psi}}$ for $0 \leq \text{Re}(z) \leq \frac{1}{2},$ and the function
\begin{equation}
	z \mapsto e^{- z K_{\Psi}} a |\Psi \rangle
\end{equation}
is holomorphic in that strip.
From this, one can easily show that the bra-valued function
\begin{equation}
	z \mapsto \langle \Psi| a e^{- z K_{\Psi}}
\end{equation}
is defined and holomorphic in the same strip.
In fact, these statements hold not just when $a$ is a bounded operator in $\A$, but also when $a$ is unbounded provided that both $a$ and its adjoint are affiliated with $\A$.
In particular, analytic continuations of this kind exist when $a$ is a smeared field operator.
For a more detailed explanation, see \cite[section 3 and appendix B]{Sorce:modular}.

\subsection{Fields, energy, and analyticity}
\label{subsec:field-analyticity}

In section \ref{sec:future-directed}, I will argue that in any QFT state that is analytically well behaved, geometric modular flow must point to the future.
The proof is inspired by results in axiomatic quantum field theory from the 1960s and 1970s, reviewed e.g. in \cite{streater2000pct}.
The details of axiomatic quantum field theory will not so important in this paper, but it will be important to understand the big conceptual lesson from that framework: that in quantum field theory, there is a deep relationship between the positivity of energy and the analytic structure of correlation functions, and that this analytic structure has physical implications.

Axiomatic quantum field theory is traditionally formulated for excitations around the vacuum state $|\Omega\rangle$ in Minkowski spacetime.
While the correlation functions
\begin{equation}
	\langle \Omega | \phi(x_1) \dots \phi(x_n) | \Omega\rangle
\end{equation}
cannot be treated rigorously as functions due to singularities, it is assumed that they can be treated rigorously as distributions, meaning that for any compactly supported smooth functions $f_1, \dots, f_n : \reals^{d+1} \to \comps,$ the smeared correlator
\begin{equation}
	\langle \Omega | \phi[f_1] \dots \phi[f_n] | \Omega\rangle
\end{equation}
depends linearly and continuously on the smearing functions $f_j$, with continuity defined in some appropriate sense.

The only other assumption we will need from the axiomatic literature is the spectrum condition, which is the statement that the vacuum is the state of lowest energy in any inertial reference frame.
To formulate this precisely, one first assumes the existence of a momentum operator $P_\mu$ that generates translations, so that for any vector $X^{\mu}$, there is a unitary operator $e^{-i P_{\mu} X^{\mu}}$ generating a translation by $X^{\mu}$.
The vacuum is assumed to be the unique state that is invariant under all translations, and it is further assumed that if $X^\mu$ is future-directed, then the operator $P_\mu X^\mu$ --- which is the Hamiltonian for the ``time direction'' $X^\mu$ --- is nonnegative.

The spectrum condition has remarkable consequences for the analytic properties of correlation functions.
The rigorous proofs of these properties can be found in \cite{streater2000pct}, and involve technical manipulations of Fourier transforms of distributions.
The basic idea, however, is simple.
We can rewrite a general correlation function in terms of the momentum operator as
\begin{equation}
	\langle \Omega | \phi(x_1) \dots \phi(x_n) | \Omega \rangle
	= \langle \Omega | \phi(0) e^{i P_{\mu} (x_2^{\mu} - x_1^{\mu})} \phi(0) \dots e^{i P_{\mu} (x_n^{\mu} - x_{n-1}^{\mu})} \phi(0) | \Omega \rangle
\end{equation}
This is a function on $n$ copies of $\reals^{d+1},$ i.e., it is a complex-valued function of $\reals^{n(d+1)}.$
It can be continued to the complex space $\comps^{n(d+1)}$ by plugging in $n$ complex vectors $z_j^{\mu} = x_j^{\mu} + i y_{j}^{\mu}$.
Naively, we only know how to define the correlation function when each $y_j^{\mu}$ vanishes.
But because the vectors $z_j^{\mu}$ only show up in exponentials of the form
\begin{equation}
	e^{i P_{\mu} (z_{j+1}^{\mu} - z_j^{\mu})}
	= e^{i P_{\mu} (x_{j+1}^{\mu} - x_j^{\mu})} e^{- P_{\mu} (y_{j+1}^{\mu} - y_{j}^{\mu})},
\end{equation}
we expect the correlation function to be well defined whenever each vector $y_{j+1}^{\mu} - y_j^{\mu}$ is future-directed, since in this case each operator $e^{i P_{\mu} (z_{j+1}^{\mu} - z_j^{\mu})}$ is bounded.
Indeed, rigorous arguments show that correlation functions can be continued to any point $(z_{1}^{\mu}, \dots, z_n^{\mu})$ in $\comps^{n(d+1)}$ such that each vector $\text{Im}(z_{j+1}^{\mu} - z_j^{\mu})$ is in the future light cone.
The continued correlation functions are holomorphic in the interior of this domain, and they limit to the standard correlation functions in a distributional sense when the imaginary parts of all vectors are taken to zero.

The analytic properties of correlation functions enforced by the spectrum condition give a precise meaning to the $i\epsilon$ prescription for resolving singularities.
As a concrete example, the two-point function
\begin{equation}
	\langle \Omega | \phi(x) \phi(y) | \Omega \rangle
\end{equation}
has singularities when $x$ and $y$ coincide or are lightlike separated, but the above considerations tell us that if a small past-directed vector is added to the imaginary part of $x,$ or a small future-directed vector is added to the imaginary part of $y,$ then these singularities are resolved, and the two-point function becomes analytic.
In any inertial coordinate frame, this can be accomplished by subtracting $i\epsilon$ from the time component of the point $x,$ or by adding $i\epsilon$ to the time component of the point $y.$

It is actually possible to make stronger statements about analyticity and the spectrum condition by considering the analytic structure of vector-valued functions.
An identical argument to the one given above shows that the function
\begin{equation}
	(x_1, \dots, x_n)
		\mapsto \phi(x_1) \dots \phi(x_n) |\Omega\rangle
\end{equation}
admits an analytic continuation to $(x_1 + i y_1, \dots, x_n + i y_n)$ whenever the vectors $y_1^{\mu}$ and $y_{j+1}^{\mu} - y_{j}^{\mu}$ are all in the future light cone.
Similar statements can be made for a non-vacuum state $|\Psi\rangle$ so long as the support of $|\Psi\rangle$ on large-momentum states is exponentially decaying.
To see this, we write
\begin{equation}
	\phi(x_1) \dots \phi(x_n) |\Psi \rangle
		=  e^{i P_{\mu} x_1^{\mu}} \phi(0) e^{i P_{\mu} (x_2^{\mu} - x_1^{\mu})} \dots e^{i P_{\mu} (x_n^{\mu} - x_{n-1}^{\mu})} \phi(0) e^{- i P_{\mu} x_n^{\mu}} |\Psi\rangle.
\end{equation}
If we replace $x_j^{\mu}$ with $x_j^{\mu} + i y_{j}^{\mu},$ then all of the exponentials except for the last one decay whenever we have $y_1^{\mu}$ and $y_{j+1}^{\mu} - y_{j}^{\mu}$ in the future light cone.
The last exponential will never decay, because we have
\begin{equation}
	y_{n}^{\mu}	
		= (y_{n}^{\mu} - y_{n-1}^{\mu}) + (y_{n-1}^{\mu} - y_{n-2}^{\mu}) + \dots + y_1^{\mu},
\end{equation}
so $y_{n}^{\mu}$ will be in the future light cone, and $e^{- i P_{\mu} (x_n^{\mu} + i y_{n}^{\mu})}$ will be unbounded.
However, if $|\Psi\rangle$ happens to be in the domain of the unbounded operator $e^{P_{\mu} y_n^{\mu}}$, then the vector $\phi(x_1 + i y_1) \dots \phi(x_n + i y_n) |\Psi\rangle$ is well defined and depends holomorphically on $x+iy$.\footnote{See \cite[section 2.4]{Sorce:modular} for a detailed explanation of the relationship between the domain of a family of unbounded operators and its analytic properties.}
So for a general vector $|\Psi\rangle,$ we conclude that the vector $\phi(x_1) \dots \phi(x_n) |\Psi\rangle$ admits an analytic continuation so long as there is an open set of vectors $y_{n}^{\mu}$ in the future light cone such that $|\Psi\rangle$ is in the domain of $e^{P_{\mu} y_{n}^{\mu}}$; in this case, the analytic continuation is defined when $y_{n}^{\mu}$ is in that open set and the vectors $y_{1}^{\mu}$ and $y_{j+1}^{\mu} - y_j^{\mu}$ are all future-directed.

The condition that $|\Psi\rangle$ is in the domain of $e^{P_{\mu} y_n^{\mu}}$ may be interpreted as the statement that $\int d p_{\mu}\, |e^{p_{\mu} y_n^{\mu}} \langle p_{\mu} |\Psi \rangle|^2$ is finite, which we may think of as the statement that the support of $|\Psi\rangle$ on large-momentum states decays exponentially.\footnote{In \cite{Strohmaier:timelike-tube}, states with this property were called ``analytic,'' and were described as states ``from which it is exponentially unlikely to extract an asymptotically large amount of energy.''
We will return to this in section \ref{sec:analytic-states}.}

\subsection{The disk method for analytic continuation}
\label{subsec:disk-method}

For holomorphic functions of a single complex variable, the question of whether a function can be analytically continued is a question about the \textit{function}, not about the domain.
For example, the function $f(z) = \frac{1}{z}$ is analytic in the domain $\comps - \{0\},$ but cannot be analytically continued to all of $\comps.$
More generally, one can show that for any nested pair of domains $\Omega \subsetneq \tilde{\Omega},$ there exists a function analytic in $\Omega$ that has no analytic continuation to $\tilde{\Omega}.$

For $n \geq 2,$ holomorphic functions from $\comps^n$ to $\comps$ have surprising properties that are not present in the case $n=1.$
In particular, there exist pairs of domains $\Omega \subsetneq \tilde{\Omega}$ such that \textit{every} function holomorphic in $\Omega$ admits a unique holomorphic extension to $\tilde{\Omega}.$
This is an extraordinary fact; it means that analytic continuation above one complex dimension is sometimes a problem of geometry, and is independent of the function you choose.
When this happens, the region $\tilde{\Omega}$ is called an \textit{envelope of holomorphy} for $\Omega.$

Many of the fundamental results of axiomatic quantum field theory are obtained by judicious application of techniques related to envelopes of holomorphy.
One begins by knowing that a correlator can be analytically continued from spacetime into a complexified domain, then uses envelope-of-holomorphy techniques to show that it can be extended much further than was initially assumed.
Because an analytic function is uniquely determined by its restriction to an arbitrarily small open set in its domain, the existence of a large domain of analyticity for a correlation function can be used to understand the physics  of quantum fields --- an identity that one wants to prove for correlation functions in spacetime might be much easier to prove in a different corner of the extended domain.

Many envelopes of holomorphy in $\comps^n$ with $n \geq 2$ are found by using a powerful technique called the disk method.
It is sketched heuristically in figure \ref{fig:disk-method}.
The basic idea is that if $\eta : \comps \to \comps^n$ is a holomorphic embedding of the complex plane into a higher-dimensional complex space, then for any holomorphic function $f : \comps^n \to \comps$, the function $f \circ \eta$ is holomorphic, and therefore satisfies the Cauchy integral formula
\begin{equation} \label{eq:disk-method-integral}
	f(\eta(z))
		= \frac{1}{2 \pi i} \int_{|\zeta| = 1} d\zeta\, \frac{f(\eta(\zeta))}{\zeta - z}. 
\end{equation}
There are many ways to embed the disk $|\zeta| \leq 1$ holomorphically in $\comps^n$.
In particular, if $f$ is only holomorphic on some subdomain of $\comps^n$, then it may happen that the image of the circle $|\zeta| = 1$ lies in the domain of $f$, while the image of the interior $|\zeta| < 1$ does not.
When this happens, equation \eqref{eq:disk-method-integral} can be used to \textit{define} a function that could plausibly be an analytic continuation of $f$; in certain circumstances, one can prove that the function so-constructed actually \textit{is} an analytic continuation.
In what follows, I will give a simple example for intuition.
The disk method will be applied to a slightly more complicated scenario in section \ref{sec:future-directed}, with details in appendix \ref{app:axis-extension}.

\begin{figure}
	\centering
	\includegraphics[scale=1.25]{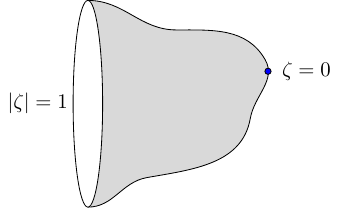}
	\caption{A two-dimensional disk embedded in a complex space $\comps^n.$
	If the embedding is a holomorphic function of the \textit{complex} disk $|\zeta| \leq 1,$ and if $f : \comps^n \to \comps$ is holomorphic, then the value of $f$ at the tip of the disk can be written as an integral over values of $f$ at the disk's edge.
	In some settings, this idea can be used to extend a function defined in a vicinity of the disk's edge all the way into the disk's interior.}
	\label{fig:disk-method}
\end{figure}

To begin, let us see why the disk method does not work in one dimension.
Suppose that $f : \comps - \{0\} \to \comps$ is a holomorphic function defined away from the origin.
We want to determine whether it has an analytic extension $\hat{f} : \comps \to \comps.$
Suppose for a moment that such an extension does exist.
Because $\hat{f}$ is holomorphic, it must satisfy the integral formula
\begin{equation}
	\hat{f}(z)
		= \frac{1}{2 \pi i} \int_{|\zeta| = |z| + \epsilon} d\zeta\, \frac{\hat{f}(\zeta)}{\zeta - z}. 
\end{equation}
The contour over which the integral is taken lies within the domain of $f$, so since $\hat{f}$ is supposed to be an extension of $f$, we may substitute $\hat{f} \to f$ in the integrand and obtain the formula
\begin{equation} \label{eq:fhat-definition}
	\hat{f}(z)
	= \frac{1}{2 \pi i} \int_{|\zeta| = |z| + \epsilon} d\zeta\, \frac{f(\zeta)}{\zeta - z}. 
\end{equation}
This is a constructive formula for $\hat{f}$, which \textit{must} give the analytic extension of $f$ if such an analytic extension exists.
But simply having a formula does not tell us that the extension exists, because we don't actually know if the function $\hat{f}$ defined by this formula is analytic, or if it is an extension of $f.$

It is in fact possible to show that the function $\hat{f}(z)$ defined by equation \eqref{eq:fhat-definition} is analytic.
It would be nice to prove this by differentiating with respect to $z$ and appealing to the theorems that tell us when a derivative and integral can be interchanged; unfortunately, because our contour $|\zeta| = |z| + \epsilon$ has non-analytic dependence on $z,$ this method does not work.
An easy workaround is to go back and remove the explicit $z$-dependence from our contour, for example, by defining a sequence of contours $|\zeta| = n$ and defining $\hat{f}(z)$ for $|z| < n$ using an appropriate contour from the sequence.

So up to checking a few small details, there is no issue with showing that $\widehat{f}(z)$ is holomorphic.
The problem is that it is not always an analytic extension of the original function $f$!
In particular, for $f(z) = \frac{1}{z},$ a straightforward calculation gives $\hat{f}(z) = 0,$ so $f$ and $\hat{f}$ are not equal anywhere in $\comps.$

For $n \geq 2,$ something different happens.
As in the case $n=1$, holomorphic functions in higher dimensions can be expressed as integrals over circles using Cauchy's integral formula; unlike in $n=1,$ however, there are additional dimensions into which those circles can be deformed.
To see how this works, let $f(z_1, z_2)$ be a holomorphic function in $\comps^2 - \{0\}.$
This means that it is continuous and is complex-differentiable in each variable independently.
As in the case $n=1,$ one can show that holomorphy is equivalent to analyticity, i.e., equivalent to the existence of a convergent power series for $f$ around every point in its domain \cite[page 30]{Vladimirov:book}.

Suppose now that $\hat{f}$ is a holomorphic extension of $f$ to all of $\comps^2$.
Because it is holomorphic in each variable, $\hat{f}(z_1, z_2)$ must satisfy Cauchy's integral formula with respect to $z_1$ and $z_2$ separately.
In particular, it must satisfy
\begin{equation}
	\hat{f}(z_1, z_2)
		= \frac{1}{2 \pi i} \int_{|\zeta| = |z_1| + \epsilon} d\zeta\, \frac{\hat{f}(\zeta, z_2)}{\zeta - z_1}.
\end{equation}
As in the case $n=1,$ the contour lies within the domain of $f,$ so we can write
\begin{equation}
	\hat{f}(z_1, z_2)
	= \frac{1}{2 \pi i} \int_{|\zeta| = |z_1| + \epsilon} d\zeta\, \frac{f(\zeta, z_2)}{\zeta - z_1}.
\end{equation}
This gives a constructive formula for the analytic extension $\hat{f}$ under the assumption that it exists.
As in the case $n=1,$ it is possible to show that $\hat{f}$ is holomorphic by removing the non-analytic $z_1$ dependence from the contour and computing derivatives.
Unlike in the case $n=1$, we can show in $n=2$ that $\hat{f}$ is actually an extension of $f$!
The reason is that for $z_2 \neq 0,$ the full image of the disk $|\zeta| \leq |z_1|+\epsilon$ is in the domain of $f.$
Since $f$ is holomorphic, this means that for $z_2 \neq 0,$ we have
\begin{equation}
	f(z_1, z_2)
		=	\frac{1}{2 \pi i} \int_{|\zeta| = |z_1|+\epsilon} d\zeta\, \frac{f(\gamma_1(\zeta), \gamma_2(\zeta))}{\zeta - z_1}
		= \hat{f}(z_1, z_2).
\end{equation}
Since $f$ and $\hat{f}$ agree on the open domain $z_2 \neq 0,$ and since they are both analytic, it follows that they agree on all of $\comps^2 - \{0\}.$

The key point in the above example was that the domain into which we wanted to extend $f$ could be foliated by one-complex-dimensional disks, each of which was holomorphically embedded in $\comps^2.$
Crucially, the boundary of each disk was in the domain of $f,$ and there was at least one member of the family for which the disk's whole interior was in the domain of $f$ as well.
In practice, when applying the disk method to analytically continue a function, the hardest part is showing that the function defined using Cauchy's integral formula is holomorphic.
In the above examples, we accomplished this by making the embedding function for the disk depend analytically on $(z_1, z_2).$
This will suffice for all applications of the disk method in this paper; for completeness, however, I will remark that it is possible to apply the disk method even when an analytically-varying family of disks cannot be explicitly constructed \cite[section 17]{Vladimirov:book}.

To my knowledge, the disk method was discovered in its most basic form by Dyson in \cite{dyson1958connection}.
The first major applications of the disk method were by Vladimirov \cite{Vladimirov:disk} and Borchers \cite{Borchers:tube}.
Many applications of the disk method can be found in \cite[chapter V]{Vladimirov:book}.

\section{Constraints on geometric modular flows}
\label{sec:constraints}

This section concerns necessary conditions for a geometric flow on spacetime to be the modular flow of a state in a quantum field theory.

In section \ref{sec:conformality}, I give a precise definition of what it means for a unitary group $e^{-i K t}$ to act geometrically on a quantum field theory, then prove that any such family of geometric transformations must be conformal.
In particular, if $\xi^a$ is a smooth vector field generating the diffeomorphisms implemented by $e^{- i K t},$ then $\xi^a$ must be a conformal Killing vector field.

In section \ref{sec:future-directed}, I show that in states with sufficiently nice analytic properties, $\xi^a$ must be directed toward the future.

In section \ref{sec:isometric}, I argue that in a non-conformal field theory, geometric modular flow should be isometric, i.e., should implement a conformal transformation with trivial conformal factor.
The arguments of section \ref{sec:isometric} are heuristic, not rigorous, but nevertheless seem worth including here.

\subsection{Geometric unitary flows must be conformal}
\label{sec:conformality}

Let $\M$ be a globally hyperbolic spacetime, and let $\H$ be a Hilbert space describing a sector of a quantum field theory.
As in section \ref{subsec:field-analyticity}, we associate to any smooth, compactly supported smearing function $f$ an unbounded operator
\begin{equation}
	\phi[f]
		= \int d^{d+1} x \sqrt{|g|} \phi(x) f(x).
\end{equation}
We will assume that there exists a dense domain of Hilbert space on which arbitrary products of fields can act; i.e., we will assume that there is a dense subspace of $\H$ that is included in the domain of every product $\phi[f_1] \dots \phi[f_n]$.
This assumption allows us to multiply field operators together, and follows naturally from basic axiomatic assumptions about quantum field theory in general backgrounds --- see e.g. \cite[page 3]{Hollands:axioms}.

Let $U$ be a unitary operator, and let $\psi$ be a diffeomorphism on spacetime.
We say that $U$ acts geometrically, implementing the diffeomorphism $\psi$, if for any smearing function $f$, the operator $U^{\dagger} \phi[f] U$ can be written as
\begin{equation}
	U^{\dagger} \phi[f] U
		= \phi[\tilde{f}],
\end{equation}
where the support of $\tilde{f}$ is the image of the support of $f$ under $\psi$, i.e.,
\begin{equation}
	\supp(\tilde{f})
		= \psi(\supp(f)).
\end{equation}
Note that this is a weaker assumption than assuming that the two functions are related by a pullback; i.e., we are not assuming $\tilde{f}(x) = f(\psi^{-1}(x)).$
This is because if $\phi$ has nontrivial dimension, then we want to allow $U$ to implement a linear transformation on the smearing function in addition to displacing it geometrically.

We will now consider a unitary group $U(t) = e^{- i K t},$ and assume that it implements a one-parameter group of diffeomorphisms $\psi_t.$
The group $\psi_t$ is generated by a vector field $\xi^a.$
In what remains of this subsection, we will show that each $\psi_t$ must be a conformal transformation, so that $\xi^a$ must be a conformal Killing vector field.
We will begin by showing that $\psi_t$ must take spacelike separated points to spacelike separated points.

To prove this, we will make two mild assumptions that should be satisfied in any reasonable quantum field theory.\footnote{As explained in section \ref{sec:assumptions}, we state everything for a theory generated by a single real scalar field; section \ref{sec:assumptions} explains how to upgrade these assumptions for general theories.}
The first is microcausality: if $f$ and $g$ are smearing functions with spacelike separated supports, then $\phi[f]$ and $\phi[g]$ commute.
The second assumption is that for any null-separated points $x$ and $y,$ we have $[\phi(x), \phi(y)] \neq 0.$
More concretely, for any regions $\Omega_1$ and $\Omega_2$ connected by a null curve, there must exist smearing functions $f_1$ and $f_2$ supported in $\Omega_1$ and $\Omega_2$ respectively, with $[\phi[f_1], \phi[f_2]] \neq 0.$

Now fix a value of $t$ and suppose, toward contradiction, that $\psi_t$ does not always take spacelike separated points to spacelike separated points.
I.e., assume there exists a pair of points $x$ and $y$ that are spacelike separated, but such that there is a causal curve connecting $\psi_t(x)$ and $\psi_t(y).$
See figure \ref{fig:microcausality} for a pictorial representation of what follows.
Because $x$ and $y$ are spacelike separated, because $\psi_t(x)$ and $\psi_t(y)$ are causally separated, and because $\psi_t$ depends continuously on $t,$ there must exist some $t_0$ between $0$ and $t$ for which $\psi_{t_0}(x)$ and $\psi_{t_0}(y)$ are null separated.\footnote{More concretely, if $\psi_t(x)$ is in the future of $\psi_t(y)$, then by continuity there exists some $t_0$ for which $\psi_{t_0}(x)$ is on the boundary of the future of $\psi_{t_0}(y),$ and in a globally hyperbolic spacetime, any point on the boundary of the future of another point can be connected to it by a null curve; see e.g. \cite[remark 3.10]{penrose1972techniques}.}
Let $\Omega_1$ and $\Omega_2$ be neighborhoods of $x$ and $y,$ sufficiently small that all points in $\Omega_1$ are spacelike separated from all points in $\Omega_2.$
The open neighborhoods $\psi_{t_0}(\Omega_1)$ and $\psi_{t_0}(\Omega_2)$ can be connected by a null curve, so by our earlier assumption, there exist smearing functions $\tilde{f}_1$ in $\psi_{t_0}(\Omega_1)$ and $\tilde{f}_2$ in $\psi_{t_0}(\Omega_2)$ with $[\phi[\tilde{f}_1], \phi[\tilde{f}_2]] \neq 0.$
But conjugating by $U$ gives $[\phi[f_1], \phi[f_2]] \neq 0$ for a pair of smearing functions $f_1$ and $f_2$ with spacelike separated support, which contradicts microcausality.

\begin{figure}[h]
	\centering
	\makebox[\textwidth][c]{\includegraphics{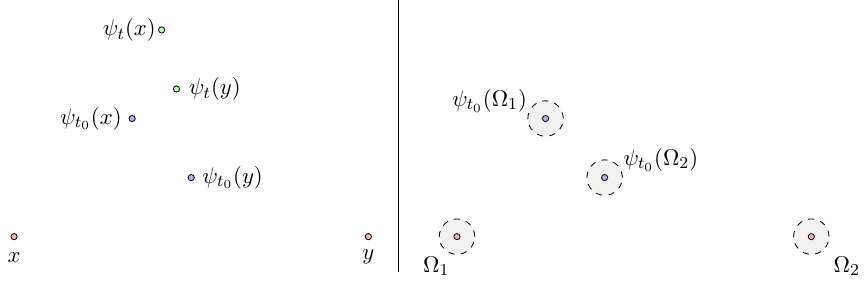}}
	\caption{\textit{Left:} A pair of spacelike separated points $x$ and $y$ that are mapped, via the diffeomorphism $\psi_{t},$ to a pair of causally separated points.
		By moving $t$ continuously toward zero, we can find a parameter $t_0$ such that $\psi_{t_0}(x)$ and $\psi_{t_0}(y)$ are null separated.
	\textit{Right:} Neighborhoods $\Omega_1$ and $\Omega_2$ of $x$ and $y$ that are small enough to be entirely spacelike separated.
	Their images under $\psi_{t_0}$ are null separated, so one can always construct a pair of non-commuting smeared fields in $\psi_{t_0}(\Omega_1)$ and $\psi_{t_0}(\Omega_2).$}
	\label{fig:microcausality}
\end{figure}

It is a general fact --- see e.g. \cite{peleska1984characterization} --- that any diffeomorphism that takes spacelike separated points to spacelike separated points is a conformal transformation.
The proof is easy, so we will give it explicitly here.
\begin{theorem}
	If $\psi$ is a diffeomorphism that takes spacelike separated points to spacelike separated points, then $\psi_* g_{ab}$ and $g_{ab}$ are related by a Weyl factor.
\end{theorem}
	\begin{proof}
	The pullback $\psi_*$ on the metric is defined via the pushforward $\psi^*$ on vectors by
	\begin{equation}
		(\psi_* g)_{ab} T^a T^b = g_{ab} (\psi^* T)^{a} (\psi^* T)^{b}.
	\end{equation}
	Since $\psi$ takes spacelike separated points to spacelike separated points, it must by continuity take null vectors to null vectors.
	This means that the metric $(\psi_* g)_{ab}$ is null on exactly the same vectors as $g_{ab}$.
	
	We may choose a local frame for the metric $g_{ab}$ containing two null vectors $k^a, \ell^a$ satisfying
	\begin{equation}
		k^a g_{ab} \ell^a = -1, \qquad k^a g_{ab} k^b = \ell^a g_{ab} \ell^b = 0,
	\end{equation}
	together with $d-1$ spacelike vectors $(s_j)^a$ satisfying
	\begin{equation}
		(s_j)^a g_{ab} (s_k)^b = \delta_{jk}, \qquad \ell^a g_{ab} (s_j)^b = k^a g_{ab} (s_j)^b = 0.
	\end{equation}
	Since the vectors $k$ and $\ell$ are null for the metric $g,$ they are also null for the metric $\psi_* g.$
	To show that $\psi_* g$ is Weyl-equivalent to $g,$ we must show that all of the above orthogonality relations are preserved in $\tilde{g}$, and that the normalization of the basis vectors changes by a constant, i.e., we must show the identities
	\begin{align}
		\ell^a (\psi_*g)_{ab} (s_j)^b
		& = k^a (\psi_* g)_{ab} (s_j)^b = 0, \\
		(s_j)^a (\psi_* g)_{ab} (s_k)^b
			& = - k^a (\psi_* g)_{ab} \ell^b \delta_{jk}.
	\end{align}
	To ease notation going forward, we will use $v \cdot w$ to denote the inner product of $v^a$ and $w^a$ in the metric $g_{ab},$ and $\bar{v \cdot w}$ to denote the inner product of the same vectors in the metric $(\psi_* g)_{ab}.$

	First we will show that the vector $(s_j)^a$ is orthogonal to $\ell^a$ and $k^a$ in the modified metric $(\psi_* g)_{ab}$.
	The vector
	\begin{equation}
		-2 \frac{\bar{s_j \cdot \ell}}{\bar{s_j \cdot s_j}}\, (s_j)^a + \ell^a
	\end{equation}
	is easily seen to be null in the metric $\psi_*g$, so it is null in the metric $g,$ which implies $\bar{s_j \cdot \ell} = 0.$
	A similar argument gives $\bar{s_j \cdot k} = 0.$
	From this, it is easy to compute that for any $(s_j)^a,$ the vector
	\begin{equation}
		(s_j)^a + k^a - \frac{1}{2} \frac{\bar{s_j \cdot s_j}}{\bar{k \cdot \ell}} \, \ell^a
	\end{equation}
	is null in the metric $\psi_* g,$ hence null in the metric $g,$ which gives the normalization condition $\bar{s_j \cdot s_j} = - \bar{k \cdot \ell}.$
	All that remains is to show $\bar{s_j \cdot s_k} = 0$ for $j \neq k,$ which is accomplished by applying the above normalization argument to the linear combination $(s_j)^a + (s_k)^a.$
\end{proof}

So far, we have shown that if the one-parameter unitary group $e^{-i K t}$ implements the one-parameter group of diffeomorphisms $\psi_t,$ then each $\psi_t$ is a conformal transformation.
In particular, this means that if $\psi_t$ is generated by the smooth vector field $\xi^a$, then $\xi^a$ is a conformal Killing vector field:
\begin{equation}
	\Del_a \xi_b + \Del_b \xi_a \propto g_{ab}.
\end{equation}
Because modular flow is a unitary group, the above argument shows that when modular flow is geometrically local, it must implement a conformal symmetry of the spacetime.
This is already a serious restriction on when geometric modular flow can arise, because a generic spacetime will not possess any conformal symmetries at all.

\subsection{Geometric modular flow in analytic states must be future-directed}
\label{sec:future-directed}

We will say that modular flow for the state $|\Psi\rangle$ is geometrically local in the region $A$, with smooth vector field $\xi^a,$ if for any smearing function $f$ supported in $A$, the operator $e^{i K_{\Psi} s} \phi[f] e^{-i K_{\Psi} s}$ can be represented as a field operator with support obtained by pushing the support of $f$ forward along $\xi^a$ for time $s$.
From subsection \ref{sec:conformality}, we know that when modular flow is geometrically local, the vector field $\xi^a$ must be a conformal Killing vector field.

In this subsection, we will show that for states with good analytic properties, geometrically local modular flow must be future-directed.
The basic reason for this is that every quantum field theory comes equipped with local positivity of energy for physical observers in some asymptotic, ultraviolet sense.
Modular flow provides another notion of energy which is similarly locally positive.
If modular flow is not future-directed, then the presence of two ``misaligned'' locally positive energies overconstrains the theory and makes it unphysical.

A proof that geometric modular flow is future-directed was given for the Minkowski vacuum by Trebels \cite{Trebels} in the framework of axiomatic quantum field theory.
In section \ref{sec:vacuum-proof}, I present a modified version of Trebels's proof; the advantage of the modified proof is that it can be generalized.
In section \ref{sec:analytic-states}, I introduce a class of ``weakly analytic'' states, and explain that this class contains the analytic states introduced in \cite{Strohmaier:analytic-states, Strohmaier:timelike-tube} by Strohmaier and Witten.
In section \ref{sec:analytic-proof}, I show that the proof of section \ref{sec:vacuum-proof} applies to any weakly analytic state, so that in these states, geometric modular flow must be directed toward the future.

\subsubsection{Proof in the Minkowski vacuum}
\label{sec:vacuum-proof}

We will work in Minkowski spacetime, $\reals^{d+1},$ and let $x^0$ be an arbitrary inertial time coordinate.
The region $A$ will be a domain of dependence of a partial Cauchy slice, with associated field algebra $\A.$
The Minkowski vacuum is cyclic and separating for $A$ --- see e.g. \cite{Reeh:1961ujh, Witten:notes} --- so there is a modular Hamiltonian $K_{\Omega}$ for the restriction of $|\Omega\rangle$ to $A$.
Assume that the associated modular flow is geometric, and is generated by the smooth vector field $\xi^a.$
Assume, toward contradiction, that there is a point $x \in A$ at which the vector field $\xi^a$ generating modular flow is not future-directed.\footnote{In particular, since the zero vector is technically future-directed, we are assuming that $\xi^a$ is nonzero at $x.$}
See figure \ref{fig:bad-vector-field}.

\begin{figure}[h]
	\centering
	\includegraphics{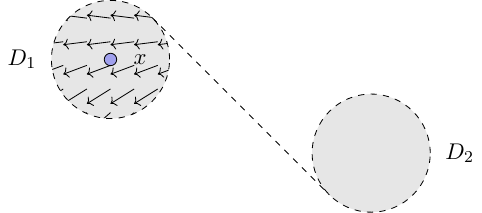}
	\caption{A neighborhood $D_1$ of the point $x$ in which $\xi$ is not future-directed, together with a neighborhood $D_2$ that is spacelike separated from $D_1 + \epsilon \xi$ and null separated from $D_1 - \epsilon \xi.$
	The region $D_1$ is spacelike separated from $D_2 + \epsilon \left( \frac{\del}{\del x^0} \right),$ and null separated from $D_2 - \epsilon \left( \frac{\del}{\del x^0} \right).$}
	\label{fig:bad-vector-field}
\end{figure}

Because $\xi$ is continuous, we can always choose a neighborhood $D_1$ of $x$ that is contained within $A$, and within which $\xi$ is not future-directed, together with a neighborhood $D_2$ that is spacelike separated from $D_1$ and from $D_1 + \epsilon \xi,$ but null separated from $D_1 - \epsilon \xi$; see again figure \ref{fig:bad-vector-field}.
In this setting, it is clear that $D_2 + \epsilon \left( \frac{\del}{\del x^0} \right)$ is spacelike separated from $D_1,$ while  $D_2 - \epsilon \left( \frac{\del}{\del x^0} \right)$ is null separated from $D_1$.

Now let $f_1$ be a smearing function in $D_1,$ and let $f_2$ be a smearing function in $D_2.$
Let $f_{2, t}$ be the smearing function obtained by shifting $f_2$ in the $x^0$ direction by time $t.$
Let $f_{1, s}$ be the smearing function for the operator $e^{i K_{\Omega} s} \phi[f_1] e^{- i K_{\Omega} s}.$
We have
\begin{equation}
	\phi[f_{1, s}] |\Omega \rangle
		= e^{i K_{\Omega} s} \phi[f_1] |\Omega \rangle.
\end{equation}
Since $\phi[f_1]$ is affiliated with $\A,$ we have from section \ref{subsec:modular-flow} that the ket-valued function
\begin{equation}
	s \mapsto \phi[f_{1, s}] |\Omega \rangle
\end{equation}
admits an analytic continuation from real values of $s$ to the strip $0 \leq \text{Im}(s) \leq \frac{1}{2},$ and the bra-valued function
\begin{equation}
	s \mapsto \langle \Omega | \phi[f_{1, s}]
\end{equation}
admits an analytic continuation to the strip $-\frac{1}{2} \leq \text{Im}(s) \leq 0.$
We also have
\begin{equation}
	\phi[f_{2, t}] |\Omega \rangle
		= \int dx\, \phi(x) |\Omega \rangle f(x - t \del_{x_0}),
\end{equation}
and from section \ref{subsec:field-analyticity}, we know that for future-pointing vectors $y,$ we may write this as
\begin{align}
	\begin{split}
	\phi[f_{2, t}] |\Omega \rangle
		& = \lim_{y \to 0} \int dx\, \phi(x + i y) |\Omega \rangle f(x - t \del_{x_0}) \\
		& = \lim_{y \to 0} \int dx\, \phi(x + t \del_{x_0} + i y) |\Omega\rangle f(x),
	\end{split}
\end{align}
with the vector $\phi(x + t \partial_{x_0} + i y) |\Omega \rangle$ depending holomorphically on $(x + t \partial_{x_0}) + i y.$
In particular, for $\tau$ positive, we have
\begin{align}
	\begin{split}
		\phi[f_{2, t}] |\Omega \rangle
		& = \lim_{\tau \to 0} \int dx\, \phi(x + (t + i \tau) \del_{x_0}) |\Omega\rangle f(x),
	\end{split}
\end{align}
which shows that $\phi[f_{2,t}] |\Omega\rangle$ may be extended holomorphically into the domain $\text{Im}(t) \geq 0.$
A similar argument shows that $\langle \Omega | \phi[f_{2, t}]$ has a holomorphic extension into the domain $\text{Im}(t) \leq 0.$

Consider now the correlation functions
\begin{equation} \label{eq:F-on-boundary}
	F(s, t)
		= \langle \Omega | \phi[f_{1, s}] \phi[f_{2, t}] | \Omega \rangle
\end{equation}
and 
\begin{equation}
	G(s, t)
	= \langle \Omega | \phi[f_{2, t}] \phi[f_{1, s}] | \Omega \rangle.
\end{equation}
From the above discussion, $F(s, t)$ is defined in the domain $\text{Im}(t) \geq 0, - \frac{1}{2} \leq \text{Im}(s) \leq 0$.
In the interior of this domain, at any fixed $t$, $F$ is holomorphic in $s,$ and at any fixed $s$, it is holomorphic in $t$; this is exactly the statement that $F$ is a holomorphic function of two complex variables.\footnote{See e.g. \cite[section 15.8]{Vladimirov:book} for a proof that this condition is equivalent to analyticity, i.e., to $F$ admitting a joint power series expansion in $s$ and $t$ within a neighborhood of every point of its domain.}
Similarly, $G(s, t)$ is defined in the domain $\text{Im}(t) \leq 0, 0 \leq \text{Im}(s) \leq \frac{1}{2},$ and is holomorphic in the interior of this domain.
For a sketch of the domains in which $F$ and $G$ are defined, see figure \ref{fig:discontinuity-domain}.

\begin{figure}[h]
	\centering
	\includegraphics{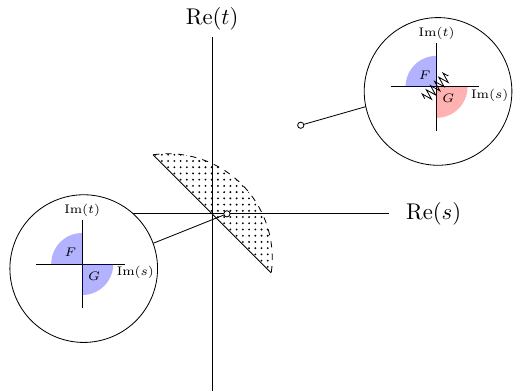}
	\caption{The real $(s, t)$ plane, together with the ``internal'' imaginary directions.
		The functions $F$ and $G$ are each defined and analytic in a portion of a certain ``internal'' quadrant.
		In the dotted region, these functions match continuously on the real $(s, t)$ plane.
		A priori, outside of the dotted region, there is a discontinuous jump from $F$ to $G$.}
	\label{fig:discontinuity-domain}
\end{figure}

It may be helpful to think of $F$ as a function defined in a certain region ``above'' the real $(s, t)$ plane --- see again figure \ref{fig:discontinuity-domain} --- and of $G$ as a function defined in a certain region ``below'' the real $(s, t)$ plane.
A priori, there is a discontinuous jump between $F$ and $G$ as one crosses the real $(s,t)$ plane.
But because of microcausality, there is at least some region in this plane where the two functions meet continuously.\footnote{Strictly speaking one can only assume continuity through a portion of the $(s,t)$ plane in a distributional sense, but there is no need to address this subtlety here, as passing from proofs under the assumption of functional continuity to proofs under the weaker assumption of distributional continuity can be accomplished in a completely standard way  --- see e.g. \cite{rudin1971lectures} or \cite[theorem 2-16]{streater2000pct}.}
For $s > 0,$ there is a small interval in $t$ around zero where $F$ and $G$ are guaranteed to agree; for $t > 0,$ there is a small interval in $s$ around zero where $F$ and $G$ are guaranteed to agree.
The domain of continuity in sketched in figure \ref{fig:discontinuity-domain}.

It turns out that any analytic function with a domain of continuity like the one sketched in figure \ref{fig:discontinuity-domain} is in fact continuous through the real $(s, t)$ plane in a much larger domain, sketched in figure \ref{fig:discontinuity-domain-extended}.
This is a consequence of the incredible power of complex analysis in dimension two.
In particular, we will soon explain how to use the disk method of section \ref{subsec:disk-method} to show that the domain of continuity for $F$ and $G$ extends all the way to a small neighborhood of $\text{Re}(s)=\text{Re}(t)=0,$ and therefore to a small interval around $\text{Re}(s)=0$ on the $\text{Re}(t)=0$ axis.
See again figure \ref{fig:discontinuity-domain-extended}.

\begin{figure}[h]
	\centering
	\includegraphics{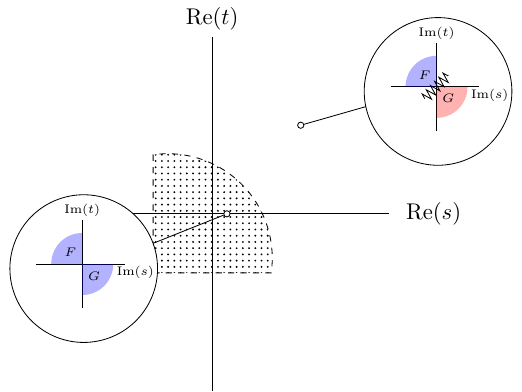}
	\caption{By applying the disk method from section \ref{subsec:disk-method}, one can show that any analytic function having a continuity region as shown in figure \ref{fig:discontinuity-domain} is also continuous through a larger region in the real plane, sketched here.
	To show this, one \textit{constructs} an analytic function that is continuous through this larger region, and shows that it is equal to the original function by matching it with the original analytic function on a portion of its domain.}
	\label{fig:discontinuity-domain-extended}
\end{figure}

Before explaining the argument for the extended domain of continuity, we will first explain why we want to prove such a thing, by showing that continuity of $F$ and $G$ through this extended domain leads to a contradiction with the assumption that modular flow is locally directed along the vector field $\xi^a.$
Continuity in the domain of figure \ref{fig:discontinuity-domain-extended} implies the identity
\begin{equation}
	\langle \Omega | \left( \phi[f_{1,s}] \phi[f_{2}] - \phi[f_{2}] \phi[f_{1,s}] \right) | \Omega \rangle
		= 0
\end{equation}
for $s$ a little bit negative.
This means that there is a commutator of null-separated operators whose vacuum expectation value vanishes.
But this argument was not actually special to the vacuum state $|\Omega\rangle.$
In particular, for any operator $O$ that is spacelike separated from $D_1$ and $D_2,$ we have $\phi[f_{2, t}] O |\Omega \rangle = O \phi[f_{2, t}] |\Omega\rangle$ for sufficiently small $t,$ and $\phi[f_{1, s}] O |\Omega \rangle = O \phi[f_{1, s}] |\Omega\rangle$ for sufficiently small $s,$ so all of the above arguments apply to the correlators
\begin{equation}
	F_{O_1, O_2}(s, t)
	= \langle O_1 \Omega | \phi[f_{1, s}] \phi[f_{2, t}] | O_2 \Omega \rangle
\end{equation}
and 
\begin{equation}
	G_{O_1, O_2}(s, t)
	= \langle O_1 \Omega | \phi[f_{2, t}] \phi[f_{1, s}] | O_2 \Omega \rangle,
\end{equation}
from which we obtain
\begin{equation}
	\langle O_1 \Omega | \left( \phi[f_{1,s}] \phi[f_{2}] - \phi[f_{2}] \phi[f_{1,s}] \right) | O_2 \Omega \rangle
	= 0
\end{equation}
for an interval of negative $s.$
A priori, different choices of $O_1$ and $O_2$ may change the domain of negative $s$ in which this equation holds, but for operators whose spacelike distance from $D_1$ and $D_2$ is bounded below, there will be an operator-independent domain in $s$ in which this equation holds.
Since the vacuum is cyclic for any local region, a dense subspace of Hilbert space can be obtained by acting with operators that are a bounded-below spacelike distance away from $D_1$ and $D_2,$ which implies the operator equation
\begin{equation}
	0 = \phi[f_{1,s}] \phi[f_{2}] - \phi[f_{2}] \phi[f_{1,s}]
\end{equation}
for sufficiently small, negative $s.$
Since this argument holds for any smearing functions $f_1$ and $f_2,$ it shows that $\phi$ commutes with itself in a distributional sense at null separation, which contradicts our basic assumptions about quantum field theory.
Once we have demonstrated the continuity shown in figure \ref{fig:discontinuity-domain-extended}, we will therefore have proven that modular flow in the Minkowski vacuum must be future-directed.

Now that we understand why we should care about the continuity of $F$ and $G$ through the domain in figure \ref{fig:discontinuity-domain-extended}, let us prove it.
We will proceed by constructing an analytic function $H$ which is continuous through the specified domain, then show that $H$ agrees with $F$ and $G$ in portions of their domains; uniqueness of analytic functions then shows that $H$ must agree with $F$ and $G$ everywhere in their domains, which implies that $F$ and $G$ are continuously connected through the region shown in figure \ref{fig:discontinuity-domain-extended}.

The construction of the function $H$ is accomplished using the disk method of section \ref{subsec:disk-method}.
Details are given in appendix \ref{app:axis-extension}.
Let us simply give a quick sketch of why the disk method will apply in this setting.
Recall that to apply the disk method, it is necessary to find, for any $z$ in the domain of extension, an analytically embedded disk in $\comps^2$ such that $z$ is in the interior of the disk and such that the boundary of the disk lies in the original domain in which $F$ and $G$ were defined.
It is also necessary to be able to deform this disk analytically through a family of disks with the same properties, such that eventually one member of the family lies entirely (boundary AND interior) within the domain where $F$ and $G$ were originally defined.

The reason this can be done to move from the domain in figure \ref{fig:discontinuity-domain} to the domain in figure \ref{fig:discontinuity-domain-extended} is a consequence of a general theorem due to Vladimirov called the C-convex hull theorem \cite[section 28]{Vladimirov:book}.
The basic idea behind the theorem is very simple, so we will explain briefly it here, though details are left to appendix \ref{app:axis-extension}.
Given any point in the domain shown in figure \ref{fig:discontinuity-domain-extended}, it is possible to place it on an analytic curve --- i.e., a parametrized curve $s(\lambda), t(\lambda)$ where both functions are analytic in $\lambda$ --- such that the endpoints of the curve lie in the original domain, and such that the tangent vectors to the curve lie entirely in the second and fourth quadrants.
See figure \ref{fig:discontinuity-domain-curve}.

\begin{figure}[h]
	\centering
	\includegraphics{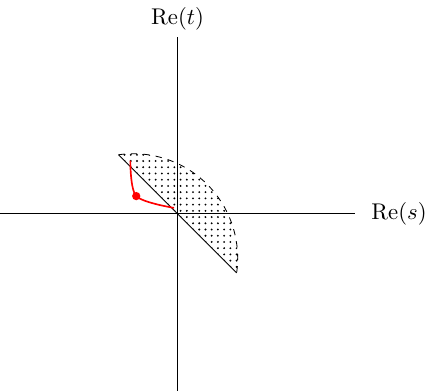}
	\caption{For any point in the dotted domain of figure \ref{fig:discontinuity-domain-extended}, it is possible to draw a curve passing through that point with endpoints in the dotted domain of figure \ref{fig:discontinuity-domain}.
	It is possible to choose this curve so that its coordinates depend analytically on its parametrization, and so that its tangent vectors are all in the second and fourth quadrants.}
	\label{fig:discontinuity-domain-curve}
\end{figure}

Such an analytic curve can be extended to a complex disk simply by replacing the real parameter $\lambda$ with a complex parameter $\zeta.$
Crucially, because the extension is analytic, \textit{the directions in which this disk extends into $\comps^2$ are completely determined by the tangent directions to the curve in the real $(s, t)$ plane.}
Concretely, the Cauchy-Riemann equations give
\begin{equation}
	\frac{\partial s}{\partial \text{Re}(\zeta)}
		= - i \frac{\partial s}{\partial \text{Im}(\zeta)},
\end{equation}
so the behavior of $s$ for small imaginary $\zeta$ is completely determined by its functional dependence on $\lambda.$
A similar statement holds for $t.$
This lets us determine the local tangent plane of the disk $(s(\zeta), t(\zeta))$ in $\comps^2.$
We have
\begin{equation}
	\left( \frac{\partial s}{\partial \text{Im}(\zeta)}, \frac{\partial t}{\partial \text{Im}(\zeta)} \right)
			= i \left( \frac{\partial s}{\partial \text{Re}(\zeta)}, \frac{\partial t}{\partial \text{Re}(\zeta)} \right).
\end{equation}
On the real $(s, t)$ plane, the vector on the right-hand side is the tangent vector to the curve shown in figure \ref{fig:discontinuity-domain-curve}; it follows that when an imaginary part is added to $\zeta,$ the disc extends outward into the imaginary version of the tangent vector to the curve.
But by construction, the tangent vector lies in the second quadrant, and its negative lies in the fourth quadrant, and these are exactly the imaginary directions in $\comps^2$ in which $F$ and $G$ are defined as analytic functions!
It follows that the boundary of the complex disk constructed from the curve in figure \ref{fig:discontinuity-domain-curve} will lie entirely in the domain of definition of $F$ and $G$.
This is exactly what was needed to apply the disk method to construct an analytic function $H$ in the domain of figure \ref{fig:discontinuity-domain-extended} that agrees with $F$ and $G$ everywhere they are defined.
Again, see appendix \ref{app:axis-extension} for a detailed version of this proof.

\subsubsection{Weakly analytic states}
\label{sec:analytic-states}

In the previous subsection, we proved that modular flow must be future-directed in the Minkowski vacuum of any quantum field theory.
But the proof we gave did not actually involve any global properties of the vacuum, such as its invariance under spacetime translations.
All we used was that in a small neighborhood of a point, the function $\phi(x) | |\Omega\rangle$ could be analytically continued into a particular complex direction in $x.$
It was not even important that we could analytically continue arbitrarily far in the complex-$x$ direction --- as we will explain further in section \ref{sec:analytic-proof}, it suffices to have an analytic continuation into an arbitrarily small domain.

Because the proof in the previous subsection was not sensitive to global details, we should expect that if we are sufficiently careful, then the same proof will apply to any state $|\Psi\rangle$ such that $\phi(x) |\Psi\rangle$ locally admits an analytic continuation of the same basic kind as when $|\Psi\rangle$ is the Minkowski vacuum.
In the Minkowski vacuum, these analytic continuations are allowed because of the positivity of global energy, as explained in section \ref{subsec:field-analyticity}.
Also in section \ref{subsec:field-analyticity}, we explained that such analytic continuations are permitted for states that have exponentially decaying support on states of asymptotically large momentum.
While this condition makes reference to a global momentum, and therefore might seem to require the global structure of Minkowski spacetime, the fact that we only care about the state's support on asymptotically \textit{large} momenta suggests that there should be a way to formulate this property even in a general curved background.
After all, we are used to thinking of large momenta as corresponding to short distances, and in short distances, every spacetime starts to look like Minkowski spacetime.

Let us begin with a very general definition of a kind of state for which, in section \ref{sec:analytic-proof}, we will see that the proof of section \ref{sec:vacuum-proof} applies.
We will explore some properties of this definition, then review the analytic states defined by Strohmaier and Witten in \cite{Strohmaier:analytic-states, Strohmaier:timelike-tube}, then finally show that every analytic state is weakly analytic.
The advantage of making this explicit connection is that in \cite{Strohmaier:analytic-states, Strohmaier:timelike-tube}, it was argued that every analytic spacetime admits a dense set of analytic states; since all such states are weakly analytic, the proof of section \ref{sec:analytic-proof} implies that in any such state, local modular flow must be future-directed.

\begin{definition}[Weakly analytic state] \label{def:weakly-analytic}
	The state $|\Psi\rangle$ will said to be \textbf{weakly analytic} if for every spacetime point $p,$ there is a compact neighborhood $S$ of $p$ and a coordinate system $x^{\mu}$ containing $S$ such that the vector-valued function
	\begin{equation}
		x^{\mu} \mapsto \phi(x^{\mu}) | \Psi \rangle
	\end{equation}
	admits an analytic continuation
	\begin{equation}
		z^{\mu} \mapsto \phi(z^{\mu}) | \Psi \rangle
	\end{equation}
	into an open set in the future imaginary light cone of nonvanishing angular width.
	
	In other words, there should exist an open cone $\Gamma$ contained in the future light cone, and an open ball $B$, such that $\phi(x^{\mu} + i y^{\mu}) | \Psi \rangle$ is defined and analytic for $x^{\mu} \in S$ and $y^{\mu} \in \Gamma \cap B.$
	See figure \ref{fig:weakly-analytic}.
\end{definition}

\begin{figure}
	\centering
	\includegraphics{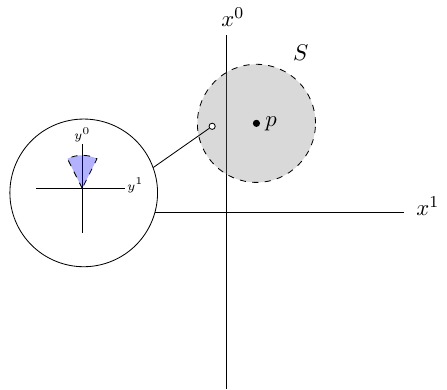}
	\caption{A coordinate system $x^{\mu}$ containing a compact neighborhood $S$ of the point $p.$
	Within this coordinate system, one can unambiguously define the future light cone in the imaginary directions $y^{\mu}.$
	The blue region shown here is a truncated cone in the ``internal,'' imaginary directions, and is contained within the imaginary future light cone.
	In a weakly analytic state $|\Psi\rangle$, every point $p$ has a compact neighborhood $S$ and a coordinate system $x^{\mu}$ such that the function $\phi(x^{\mu} + i y^{\mu}) |\Psi\rangle$ is defined and analytic for $x^{\mu}$ in $S$ and $y^{\mu}$ in such a truncated cone.}
	\label{fig:weakly-analytic}
\end{figure}

We should not expect every state to be weakly analytic, but we emphasize that it is a very weak condition.
We do not even need to be able to define the imaginary future light cone in a covariant way, because we only require that there exists a \textit{single} coordinate neighborhood of any given point in which an analytic continuation into the imaginary future light cone exists, and within any given coordinate system, the notion of an imaginary future light cone is unambiguous.
Furthermore, we only require the analytic continuation to exist for an arbitrarily small radius in some arbitrarily small imaginary solid angle.

In thinking about definition \ref{def:weakly-analytic}, it may be helpful to know that if an appropriate analytic continuation exists in one coordinate system, then it exists in any other coordinate system related to the first one by an analytic transition map.
In other words, suppose that $x^{\mu}$ and $u^{\mu}$ are two coordinate patches containing $p$, and that the transition map can always be written locally as a power series
\begin{equation}
	u^{\mu}(x)
		= \sum_{n_0 \dots n_{d}} a^{\mu}_{n_0 \dots n_{d}} (x^0 - x_*^0)^{n_0} \dots (x^{d} - x_*^{d})^{n_d}.
\end{equation}
Suppose also that $|\Psi\rangle$ is weakly analytic at the point $p$ with respect to the coordinate system $x^{\mu},$ so that there is a compact neighborhood $S$ of $p$ and an analytic extension $\phi(x^{\mu} + i y^{\mu}) | \Psi \rangle$ for $x^{\mu}$ in a compact coordinate set $S_x$ and $y^{\mu}$ in some truncated coordinate cone $\Gamma \cap B.$
Because the transition map from $x^{\mu}$ to $u^{\mu}$ is defined as a power series in a neighborhood of every point, it has a canonical extension to complex values of $x^{\mu}$ with small imaginary part, obtained simply by plugging $x^{\mu} + i y^{\mu}$ into the power series.
The analytic function $\phi(x^{\mu} + i y^{\mu}) | \Psi \rangle$ can be pushed through this map to obtain an analytic function $\phi(u^{\mu} + i v^{\mu}) | \Psi \rangle$ defined in the image of the domain $S_x + i (\Gamma \cap B).$
The trouble is that the image of this domain is not necessarily of the same form, i.e., while we certainly have $u(S_x) = S_u$, with $S_u$ a compact set in the $u$ coordinates, we do not necessarily have $u(S_x + i (\Gamma \cap B)) = S_u + i (\Gamma' \cap B')$ for some truncated cone $\Gamma' \cap B'$ contained within the future light cone.
However, one can very generally show that $u(S_x + i (\Gamma \cap B))$ will \textit{contain} a set of the form $S'_u + i (\Gamma' \cap B')$, where $S'_u$ is a compact neighborhood of $p$.
This follows from the fact that the linear term in the power series for the transition map takes $\Gamma$ into the interior of the future light cone of the $u^{\mu}$ coordinates; by shrinking $\Gamma$ and $B$ we can guarantee that the higher-order terms are not strong enough to make the imaginary part of $u(S_x + i (\Gamma \cap B))$ leave a fixed solid angle in the future light cone of the $u$ coordinates, and by shrinking $S_x$, we can guarantee that the higher-order terms are not strong enough to make the real part of $u(S_x + i (\Gamma \cap B))$ leave $S_u.$
From this, it is clear that $u(S_x + i (\Gamma \cap B))$ contains an open set of points whose projection onto the real space is a compact subset of $S_u,$ and whose projection onto the imaginary space is contained in the future light cone of the $u$ coordinates.
Any point $u_0 + i v_0$ in this set corresponds to a point $x_0 + i y_0$ in $S_x + i (\Gamma \cap B)$; by inverting the transition map and applying a similar ``shrinking'' argument, it isn't so hard to convince yourself that for $v_0$ sufficiently small, the points $u_0 + i \lambda v_0$ map to $S_x + i (\Gamma \cap B)$ for $0 \leq \lambda \leq 1,$ and that this remains true in a small solid angle around $v_0.$
A careful application of this argument shows that there is a truncated cone $\Gamma' \cap B'$ in the $u$ coordinates such that $u(S_x') + i (\Gamma' \cap B')$ is in the domain where the analytic continuation of $\phi(u^{\mu}) | \Psi \rangle$ is defined.

The above considerations tell us that if a state is weakly analytic with respect to one coordinate system, then it is weakly analytic with respect all other analytically related coordinate systems.
In other words, weak analyticity of a state could be defined with respect to a given \textit{analytic structure}, which is a preferred family of coordinate charts with analytic transition maps.
An \textit{analytic spacetime} is a spacetime that comes with such a preferred family of coordinate charts, together with a metric whose components in these charts are analytic functions.
So while in principle a weakly analytic state could exist in any smooth spacetime --- since all that is required is that there exist a single coordinate chart in which an appropriate analytic continuation exists --- it is most natural to consider them in analytic spacetimes, where the condition of weak analyticity will be satisfied or violated in every allowed coordinate system at once. 

Recently, Strohmaier and Witten \cite{Strohmaier:analytic-states, Strohmaier:timelike-tube} have studied a special class of states in analytic spacetimes, called \textit{analytic states}.
The condition for a state to be analytic is formulated using microlocal analysis, following earlier work in \cite{Radzikowski:microlocal, Brunetti:microlocal, Hollands:microlocal, Strohmaier:Hilbert}.
The idea of microlocal analysis is to study the Fourier transform of a function on a manifold in an arbitrarily small neighborhood of a given point.
Since taking a Fourier transform requires choosing a system of coordinates, Fourier transforms on manifolds are ill defined.
There is, however, universal data about the local behavior of a Fourier transform that can be extracted in a coordinate-independent way.
Below I will sketch the big-picture outlines of microlocal analysis, but it is not at all important for the present paper to understand any details; I include the following paragraph only for completeness.\footnote{Much more information about microlocal analysis can be found in \cite[chapter 8]{hormander2003analysis}; see also \cite{Strohmaier:timelike-tube} for an accessible introduction to the basic ideas .}

Given a point $p$ in the domain of a function $f,$ given a coordinate system $x^{\mu}$ containing $p,$ and given a compactly supported smooth function $\chi$ whose support contains $p$ and is contained in this coordinate patch, one can define a Fourier transform
\begin{equation}
	\widehat{\chi f}_{x}(\xi)
		= \int d^{d+1} x\, e^{- i \xi_{\mu} x^{\mu}} \chi(x^{\mu}) f(x^{\mu}).
\end{equation}
If we were to choose a different coordinate system $y^{\mu}$ and a different cutoff function $\eta$, we would find that the Fourier transform $\widehat{\eta f_y}$ is very different from the Fourier transform $\widehat{\chi f_x}$.
But we might expect that there is some data about the asymptotic behavior of the Fourier transform --- i.e., about the large-$\xi_\mu$, small-$\chi$-support behavior of $\widehat{\chi f}_x(\xi)$ --- that is independent of our choice of coordinates and of the details of our cutoff.
The specific information that is traditionally studied in microlocal analysis is the rate of growth of the Fourier transform in the direction $\xi_{\mu}$; any direction in which the Fourier transform decays faster than any rational function is said to be \textit{microlocally smooth}, and one can show (see e.g. \cite[theorem 8.2.4]{hormander2003analysis}) that the set of microlocally smooth directions of a function at a point is well defined as a subset of cotangent space, i.e., that the sets of microlocally smooth directions in two different coordinate systems are related by the usual rules for transforming cotangent vectors.
If the cutoff functions used in a particular coordinate system are analytic instead of compactly supported --- for example, if the cutoff functions are Gaussian ---  then one can restrict further to those directions in momentum space where the Fourier transform decays exponentially in the size of the cutoff function; these directions are said to be \textit{microlocally analytic}, and one can show (with a great deal of work) that the set of microlocally analytic directions at a point is well defined as a cotangent vector so long as one restricts to analytic changes of coordinates.
In other words, the set of microlocally smooth directions at a point is well defined on a smooth manifold, and the set of microlocally analytic directions is well defined on an analytic manifold.  

Many authors --- see e.g. \cite{Radzikowski:microlocal, Brunetti:microlocal, Hollands:microlocal} --- have proposed restrictions on physical quantum field theory states in curved spacetime that take the form of restrictions on which directions in cotangent space are microlocally smooth or microlocally analytic.
These conditions aim to capture some aspect of the global statement in Minkowski spacetime that the energy is nonnegative in every inertial reference frame, which can be interpreted as a restriction on where the Fourier transform of a correlation function is allowed to be nonzero.
In \cite{Strohmaier:analytic-states, Strohmaier:timelike-tube}, Strohmaier and Witten proposed a very weak version of microlocal analyticity that includes all previous proposals, and which suffices to prove interesting, general statements.
Their definition is as follows.\footnote{It is perhaps worth noting that it was shown in \cite[proposition 2.6]{Strohmaier:Hilbert} that conditions on the microlocally analytic directions of the vector $n$-point function $\phi(x_1) \dots \phi(x_n) | \Psi \rangle$ can be expressed in terms of conditions on the microlocally analytic directions of the $2n$-point correlation function $\langle \Psi | \phi(x_1) \dots \phi(x_{2n}) | \Psi\rangle$, so readers uncomfortable with imposing conditions on vector-valued functions can translate everything into statements about correlators if they like.}
\begin{definition}[Analytic state]
	A state $|\Psi\rangle$ of a QFT in a real analytic spacetime is said to be an \textbf{analytic state} if the vector $n$-point function
	\begin{equation}
		\phi(x_1) \dots \phi(x_n) | \Omega \rangle
	\end{equation}
	is microlocally analytic at $(x_1, \dots, x_n)$ in the cotangent direction $(k_1)_a, \dots, (k_n)_a$ whenever the rightmost nonzero $k$ lies outside of the dual cone of the future light cone, i.e., whenever the rightmost nonzero $k$ has $k_a v^a < 0$ for some  $v$ in the future light cone at the corresponding spacetime point.
	
	In particular, this requires that
	\begin{equation}
		\phi(x) |\Psi \rangle
	\end{equation}
	is microlocally analytic at $x$ in the cotangent direction $k_a$ so long as $k^a$ is not past-directed.\footnote{In the notational conventions of Strohmaier and Witten, $k^a$ was required to not be \textit{future}-directed.
	Interestingly, in \cite{Strohmaier:analytic-states}, this was because the ``mostly minuses'' metric signature was chosen, and in \cite{Strohmaier:timelike-tube}, it was because the ``mostly pluses'' metric signature was chosen but a nonstandard sign was chosen for the definition of the Fourier transform.}
\end{definition}

How does this definition of an analytic state in terms of microlocal analysis relate to our definition \ref{def:weakly-analytic} of a weakly analytic state in terms of analytic continuations?
The answer is that there is a very general connection between microlocal smoothness and analytic continuability, and in fact many of the statements about the cotangent directions in which a function is microlocally smooth can be equivalently stated as restrictions on the tangent directions in which the function can be analytically continued.
In particular, the Strohmaier-Witten restriction on the vector one-point function $\phi(x) | \Psi \rangle$ is equivalent to a stronger version of our definition \ref{def:weakly-analytic}.
We state the following theorem without proof; it is a special case of \cite[theorems 8.4.8 and 8.4.15]{hormander2003analysis}.\footnote{While technically the theorems in that reference are stated for scalar functions rather than Hilbert space-valued functions, there is no issue generalizing the proofs to Hilbert space; this sort of generalization was explored at length in \cite{Strohmaier:Hilbert}.}
\begin{theorem} \label{thm:analytic-is-weakly-analytic}
	Let $\M$ be a Lorentzian spacetime with a choice of analytic structure, so that microlocal analyticity is well defined.
	The function $x \mapsto \phi(x) | \Psi \rangle$ is microlocally analytic for all $k_a$ outside the dual cone of the future light cone if and only if for any coordinate system $x^{\mu},$ every compact set $S$ contained in the coordinate patch, and every open cone $\Gamma$ with closure contained in the future light cone, there exists 
	an open ball $B$ and an analytic continuation of $\phi(x) |\Psi \rangle$ to $\phi(x^{\mu} + i y^{\mu}) |\Psi\rangle$ for $x^{\mu}$ in $S$ and $y^{\mu}$ in $\Gamma \cap B.$
\end{theorem}
Our definition \ref{def:weakly-analytic} of a weakly analytic state was motivated by wanting to consider states for which a version of the $i\epsilon$ prescription works outside of Minkowski spacetime.
We said that a state was weakly analytic if the vector one-point function $\phi(x) | \Psi \rangle$ could be made fully analytic by adding a small future-directed imaginary vector to $x$.
Strohmaier and Witten, by contrast, have said that a state is analytic if every cotangent vector with negative overlap with some future-directed tangent vector is a direction of microlocal analyticity.
What we have pointed out in theorem \ref{thm:analytic-is-weakly-analytic} is that every analytic state is weakly analytic.
For our purposes, we are happy to simply assume weak analyticity as a natural condition on states, and prove that modular flow must be future-directed on states satisfying the condition.
It is good to know, however, that this condition is implied by a previously studied condition, and which was argued in \cite{Strohmaier:analytic-states, Strohmaier:timelike-tube} to be generic for quantum fields propagating on an analytic spacetime.

\subsubsection{Proof for weakly analytic states}
\label{sec:analytic-proof}

The proof presented in section \ref{sec:vacuum-proof} for the Minkowski vacuum actually works for any weakly analytic state.
This means that in a weakly analytic state, local modular flow must be future-directed.
Here we spell this out explicitly.

Let $\M$ be a general spacetime, and suppose that $A$ is a spacetime region with associated von Neumann algebra $\A$.
Let $|\Psi\rangle$ be a weakly analytic state (definition \ref{def:weakly-analytic}), and suppose that $|\Psi\rangle$ is cyclic and separating for $\A.$
Let $K_{\Psi}$ be the associated modular Hamiltonian.
Suppose that the modular flow generated by $K_{\Psi}$ is local, and implements the diffeomorphisms generated by a smooth vector field $\xi^a.$
Suppose, toward contradiction, that there is a point $x \in A$ at which $\xi^a$ is not future-directed.

Since $|\Psi\rangle$ is weakly analytic, there exists a compact neighborhood $S$ of $x$ and a coordinate system $u^{\mu}$ containing $S$, together with an open cone $\Gamma$ in the future light cone of the coordinate system $u^{\mu}$, and an open ball $B$, such that there is an analytic function $\phi(u^{\mu} + i v^{\mu})$ for $u^{\mu} \in S$ and $v^{\mu} \in \Gamma \cap B.$
We may always take $S$ to lie within the region $A$ by shrinking it, since this will not affect the existence of the analytic continuation.
Within $S$, one can always pick two domains $D_1$ and $D_2$ as in figure \ref{fig:bad-vector-field}.
The properties of these domains are:
\begin{itemize}
	\item $D_1$ contains $x$.
	\item $\xi^a$ is not future-directed anywhere in $D_1.$
	\item $D_1$ and $D_2$ are spacelike separated, but become null separated if $D_1$ is pulled back slightly along the vector field $\xi^a,$ or if $D_2$ is pulled slightly to the past in any timelike direction.
\end{itemize}
We should also choose $D_1$ and $D_2$ such that their closures are contained within the interior of $S$, so that we do not leave $S$ under small perturbations of the regions $D_1$ and $D_2.$

Because everything is taking place within the region $S$, where the coordinate patch $u^{\mu}$ is defined and where $|\Psi\rangle$ is weakly analytic, there exists a future-directed timelike vector $e^{\mu}$ in $\Gamma$ such that for any point $u_0^{\mu}$ in $S$, the vector-valued function 
\begin{equation}
	\phi(u_0^{\mu} + (t + i \tau) e^{\mu}) | \Psi\rangle
\end{equation}
is analytic in $t + i \tau$ so long as $|t|$ is sufficiently small that $u_0^{\mu} + t e^{\mu}$ remains in $S$, and so long as $\tau \geq 0$ with $\tau$ sufficiently small that $\tau e^{\mu}$ is in the truncated cone $\Gamma \cap B.$
This tells us that if $f_2$ is a smearing function in $D_2,$ and $f_{2, t}$ is the smearing function obtained by pushing $f_2$ forward along $e^{\mu}$ for time $t,$ then the vector-valued function
\begin{equation}
	\phi[f_{2, t}] | \Psi \rangle
\end{equation}
can be analytically continued to a small domain of complex values of $t$ in which $\text{Im}(t)$ is nonnegative.
If $f_1$ is a smearing function in $D_1$, and $f_{1, s}$ is the smearing function obtained via the identification $e^{i K_{\Psi} s} \phi[f_1] e^{-i K_{\Psi} s} = \phi[f_{1, s}],$ then the properties of modular flow discussed in section \ref{subsec:modular-flow} guarantee that the vector-valued function
\begin{equation}
	\phi[f_{1, s}] |\Psi \rangle
\end{equation} 
can be analytically continued to complex values of $s$ for $0 \leq \text{Im}(s) \leq \frac{1}{2}.$
For any operators $a'_1, a'_2$ in $\A',$ the arguments made in \ref{sec:vacuum-proof} then tell us that the function
\begin{equation}
	F_{a'_1, a'_2}(s, t) = \langle a'_1 \Psi | \phi[f_{1, s}] \phi[f_{2, t}] | a'_2 \Psi \rangle 
\end{equation}
can be analytically continued to a domain in the regime $\text{Im}(s) \leq 0, \text{Im}(t) \geq 0$, with this domain being independent of $a_1'$ and $a_2'$.
Similarly, the function
\begin{equation}
	G_{a'_1, a'_2}(s, t) = \langle a'_1 \Psi | \phi[f_{1, t}] \phi[f_{2, s}] | a'_2 \Psi \rangle 
\end{equation}
can be analytically continued to an $(a'_1, a'_2)$-independent domain in the regime $\text{Im}(s) \geq 0, \text{Im}(t) \leq 0.$
Furthermore, these functions agree in a regime like the one sketched in figure \ref{fig:discontinuity-domain}.
By the application of the disk method sketched in section \ref{sec:vacuum-proof} and detailed in appendix \ref{app:axis-extension}, it follows that $F_{a'_1, a'_2}$ and $G_{a'_1, a'_2}$ agree for $t=0$ and $s$ slightly negative.
Because $|\Psi\rangle$ is separating for $\A,$ it is cyclic for $\A',$ and vectors of the form $a' |\Psi \rangle$ are dense in Hilbert space.
This gives
\begin{equation}
	\left[ \phi[f_{1, s}], \phi[f_{2}] \right] = 0
\end{equation}
for an arbitrary pair of smearing functions $f_{1, s}$ and $f_2$ in null-separated regions, contradicting our fundamental assumption about quantum field theory, and proving that local modular flow in weakly analytic states is future-directed.

\subsection{Geometric modular flow in non-conformal theories is probably isometric}
\label{sec:isometric}

In section \ref{sec:conformality}, we saw that geometrically local modular flow must be generated by a conformal Killing vector field $\xi^a.$
Here, we will briefly explore whether this condition is constrained further in a theory that does not possess conformal symmetry.

It is a little difficult to specify exactly what it means for a quantum field theory in a curved background to be conformally invariant.
There is no difficulty for classical field theories, because one can simply declare a classical field theory to be one whose action is invariant under a Weyl rescaling of the metric.
Quantum field theories, however, have anomalies, and these anomalies typically depend on curvature.
What one really wants is for there to be a basis for the fields in the theory --- the so-called ``primary operators'' --- whose correlation functions transform in a predictable way when the metric is rescaled.
Assuming the existence of a stress tensor whose trace generates Weyl transformations, conformality of a field theory is typically described as the statement that the trace of the stress tensor vanishes in flat spacetime, and that in curved backgrounds, the trace is a well behaved functional of local curvature invariants.
What is meant by ``well behaved'' can depend both on spacetime dimension and on the author, but some understanding can be developed by considering metric perturbations around flat spacetime, or by studying Riemannian field theories with a Weyl-covariant path integral description.

What we will point out here is simply that if local modular flow in the region $A$ has a nontrivial conformal factor --- i.e., if it is generated by a conformal Killing vector field $\xi^a$ that fails to be isometric at some point $p$ in $A$ --- then an appropriately chosen stress tensor will have its trace at $p$ be independent of all fields, i.e., dependent only on curvature.
This is a generic feature of conformal theories, and seems like it would be hard to reproduce in any consistent non-conformal quantum theory.

The idea is quite simple.
Suppose that $A$ is a domain of dependence that is kept invariant by the conformal Killing vector field $\xi^a,$ and suppose that $|\Psi\rangle$ is a state whose modular flow in $A$ implements the conformal symmetry generated by $\xi^a.$
Since $\xi^a$ is a conformal Killing vector field, there is a smooth function $\kappa$ on $A$ satisfying
\begin{equation}
	\nabla_a \xi_b + \nabla_b \xi_a = \kappa g_{ab}.
\end{equation}
The role of a stress tensor is to generate local symmetries of fields.
Modular flow is a symmetry, so if it is local, it makes sense to think that it should be generated by a stress tensor.
More plainly, it seems reasonable to expect that there is an appropriate choice of stress tensor $T_{ab}$ such that for any Cauchy slice $\Sigma$ of $A$, the ``operator''
\begin{equation}
	\int_{\Sigma} T_{ab} \xi^a d \Sigma^b
\end{equation}
generates modular flow for fields in $A$, i.e., satisfies an equation like
\begin{equation}
	\left[ \int_{\Sigma} T_{ab} \xi^a d \Sigma^b, \phi(x) \right]
		\sim \mathscr{L}_{\xi} \phi(x).
\end{equation}
Since modular flow is supposed to be local everywhere in $A$, this expression should be formally independent of the chosen slice $\Sigma$.
Using Stokes's theorem to write the difference between two codimension-1 integrals of the stress tensor as a codimension-0 integral of the stress tensor, and using the assumption that the stress tensor is conserved, we obtain the expression
\begin{equation}
	\left[ \int_{B} \kappa (g^{ab} T_{ab}), \phi(x) \right] = 0
\end{equation}
for any codimension-0 spacetime region $B$.
Anywhere that $\kappa$ is nonzero, we conclude that the trace of the stress tensor should commute with all fields.
So if $\kappa$ is nonzero, then at the corresponding point, the stress tensor generating modular flow must be independent of the fields, which suggests that the field theory must be conformal.

\section{Path-integral construction of states with geometric modular flow}
\label{sec:sufficiency-construction}

In the previous section, we gave a few conditions that must be satisfied by any local modular flow: (i) the geometric transformation must be conformal; (ii) if the state is ``weakly analytic,'' then the geometric transformation must be future-directed; (iii) if the theory does not possess conformal symmetry, then the geometric transformation probably needs to be an isometry.
These are all \textit{necessary} conditions on local modular flow; we have no guarantee that they are sufficient.\footnote{In fact, I think there is probably an additional necessary condition that has not yet been established, which is discussed in section \ref{sec:discussion}.}
A question we would like to answer is: given a quantum field theory in a curved spacetime, a domain of dependence $A$, and a vector field $\xi^a$ in $A$ satisfying all the conditions of section \ref{sec:constraints}, does there exist a state whose modular flow implements $\xi^a$?

This question seems quite hard to answer in general, because constructing states in a general quantum field theory is a very difficult task.
The only tool I know of for constructing states in an interacting field theory is the Euclidean path integral.
If $\Sigma$ is a slice of a Lorentzian spacetime, and $\M_R$ is a Riemannian manifold with boundary isometric to $\Sigma,$ then a path integral over $\M_R$ with open boundary conditions constructs a functional of field configurations on $\Sigma$, which one can try to interpret as a state of the Lorentzian theory.
This construction is in fact used to provide a heuristic argument for the modular flow of the Minkowski vacuum in the Rindler wedge being a boost \cite{unruh1984acceleration}.
Basically, one treats the vacuum state $|\Omega\rangle$ as the path integral on a Euclidean half-plane with open boundary conditions --- see figure \ref{fig:unruh-path-integral} --- and represents the ``density matrix'' of this state in the Rindler wedge $A$ as the path integral over a full Euclidean plane with a cut along one radius --- see figure \ref{fig:unruh-path-integral}.
Because the Euclidean plane possesses rotational symmetry, which is a symmetry of the Euclidean field theory, the path integral can be ``re-sliced'' in an angular direction (as in figure \ref{fig:unruh-path-integral}) to provide the heuristic equation
\begin{equation}
	\rho_{A}
		\sim e^{- 2 \pi K},
\end{equation}
with $K$ the generator of Euclidean rotations.
Under Wick rotation, the generator of Euclidean rotations becomes the generator of Lorentz boosts, which suggests that the corresponding boost, appropriately normalized, should be the modular Hamiltonian of $|\Omega\rangle$ in the Rindler wedge.

\begin{figure}
	\centering
	\includegraphics{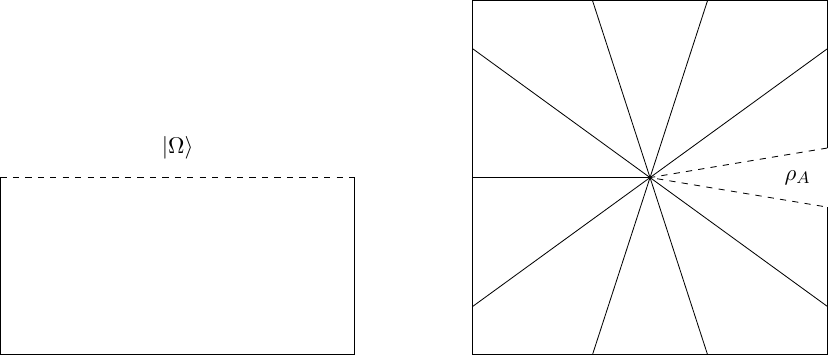}
	\caption{\textit{Left}: The Euclidean path integral preparing the Minkowski vacuum.
	\textit{Right:} The Euclidean path integral preparing the ``density matrix'' of the Minkowski vacuum in the Rindler wedge, re-sliced to show the identity $\rho_{A} \sim e^{- 2 \pi K}$ with $K$ the generator of Euclidean rotations.}
	\label{fig:unruh-path-integral}
\end{figure}

This argument can be made more precise using the work of Bisognano and Wichmann \cite{bisognano1976duality}, where the informal path integral construction of the vacuum state is replaced by theorems concerning the analytic structure of correlation functions in Minkowski spacetime.
In this argument, however, it is still the case that the existence of an analytic continuation to Euclidean signature plays a crucial role.

The path integral argument can be extended to any spacetime with a static Killing horizon that admits an analytic continuation to a complete Riemannian section --- see e.g. \cite{jacobson1994note}.
I will sketch here an argument similar to the ones presented in \cite{jacobson1994note}, which applies in certain analytic spacetimes with hypersurface-orthogonal conformal Killing vector fields.
Because the spacetime metric is presumed analytic, it can be analytically continued from any domain of dependence $A$ into a complex spacetime $A_{\comps}$ that contains $A$ as a slice.
There is a coordinate-independent notion of ``multiplication by $i$'' which takes vectors tangent to $A$ to complex vectors tangent to $A_{\comps}$.
Let $\xi^a$ be a hypersurface-orthogonal conformal Killing vector field in $A$, and let $\Sigma$ be a slice of $A$ that is orthogonal to $\xi^a.$
We can decompose $\xi^a$ uniquely into a vector $\xi_{\Sigma}^a$ that is tangent to $\Sigma,$ and a vector $\xi_{\perp}^a$ that is orthogonal to $\Sigma.$
We can then construct a new vector $\xi_{\text{Riem}} = i \xi_{\perp}^a + \xi_{\Sigma}^a$ that points locally into a Riemannian direction within $\Sigma_\comps.$
For certain well behaved spacetime regions $A$, there will be a whole Riemannian submanifold of $A_{\comps}$ such that (i) $\Sigma$ is the boundary of the submanifold, and (ii) $\xi_{\text{Riem}}^a$ can be extended into a conformal Killing vector field of the submanifold.
Given a path integral formulation of a conformal field theory, the path integral over this Riemannian manifold produces a state of the Lorentzian theory on $\Sigma$, and one can give a heuristic argument like the one in figure \ref{fig:unruh-path-integral} that the resulting state will have $\xi^a$ as its modular flow.
If $\xi^a$ was originally a Killing vector field, and if $\xi_{\text{Riem}}^a$ can be extended to a Riemannian submanifold of $A_{\comps}$ as a Killing vector field, then one expects the above construction to work in every quantum field theory, not just those with conformal symmetry.

The existence of a regular Riemannian manifold whose only boundary is $\Sigma,$ and on which $i \xi^a$ can be extended as a conformal Killing vector field, is obviously very special.
While every analytic metric can be continued into a complex metric, the domain of analytic continuation will generally be very small.
In this setting, it is hard to imagine using path integrals to construct states with local modular flow.
Progress on this general problem will likely need to begin in free theories, where there is better control over the space of states.

\section{Comments and discussion}
\label{sec:discussion}

The main purpose of this paper was to prove two general constraints on the settings in relativistic quantum field theory in which modular flow can be geometrically local.
I proved (i) that geometric modular flow must always implement a conformal symmetry of the background, and (ii) that in certain analytically well behaved states, geometric modular flow must be future-directed.
I also argued (iii) that in non-conformal theories, modular flow must implement an isometry of the spacetime.
Finally, I explained certain special settings in which path integrals can be used to construct states with local modular flow in general interacting theories.

As I see it, there are three open problems that must be solved to declare our understanding of local modular flow complete.
These are (i) the determination of additional necessary conditions for modular flow to be geometric, (ii) a version of the results of section \ref{sec:future-directed} that does not require analyticity, and (iii) a general construction, given a vector field satisfying certain sufficient conditions, of a state whose modular flow implements the corresponding diffeomorphism.
I will address now address each of these in its own subsection.

\subsection{Scaling near the edge}

My guess is that in addition to the three necessary conditions for local modular flow that were discussed in this paper, there is a fourth that restricts the parametrization of the group of diffeomorphisms, or equivalently that restricts the overall scale of the generating vector $\xi^a$.
This restriction should apply in any region $A$ that has a spacelike boundary --- i.e., in any region with an ``edge.''
In other words, if $\xi^a$ is a future-directed conformal Killing vector that preserves such a region, then there should be at most one constant $c$ such that $c \xi^a$ generates the modular flow of a quantum field theory state.
My reason for thinking this is that the modular flow of the Minkowski vacuum in the Rindler wedge is generated by
\begin{equation}
	2 \pi \left[ x \left( \frac{\del}{\del t} \right)^a + t \left( \frac{\del}{\del x}\right)^a \right],
\end{equation}
and the $2\pi$ prefactor is very important.
In the path integral derivation of \cite{unruh1984acceleration}, the factor of $2\pi$ arises from the requirement that the ``Euclidean section'' preparing a state should not have a conical defect.
In the axiomatic derivation of \cite{bisognano1976duality}, the factor of $2\pi$ arises from the relationship between the $\mathsf{CRT}$ operator and the complex Lorentz group.

In local coordinates near a point on the edge of $A$, any timelike conformal Killing vector field that vanishes at that point can be written as some linear combination of $x \del_t$ and $t \del_x$, plus terms that are higher order in $t$ and $x$.
The coefficients of the linear combination are not fixed by conformality, and can depend on tangent directions to the edge.
But due to the physical principle that ``all states should look like the vacuum at short distances,'' it seems like there should be a way to fix these coefficients to $2\pi$.
Unfortunately, a general argument seems elusive.
Any complete argument would almost certainly produce the factor of $2\pi$ by demanding consistency between the analytic continuation of correlators to Euclidean signature, and the analytic continuation of modular flow to complex values of modular time.
The trouble is that for a general conformal Killing vector field in a general region $A$, the analytic continuation to Euclidean signature is hard to control, because even if the vector field itself is controlled as one approaches the edge of $A$, analytic continuations are very sensitive to small perturbations.

It seems very important to understand whether there is a fundamental restriction on the scaling of local modular flow, in part because of an apparent connection to the long-standing but unproven conjecture that modular flow in any quantum field theory state becomes \textit{approximately} local as one approaches the edge of the region in which it is defined.\footnote{This is known to hold in a precise sense for certain regions of Minkowski spacetime thanks to a result by Fredenhagen \cite{Fredenhagen:1984dc}.
See e.g. \cite[section 4.1]{Jensen:2023yxy} for a recent discussion of the general problem.}
While there are good reasons to believe that this statement is true, it is not even clear in what sense the approximation is supposed to hold; understanding the behavior of exactly-local modular flow near an edge seems like it could be an important step in understanding the behavior of approximately-local modular flow in the same regime.

\subsection{Doing away with analyticity}
\label{sec:conclusion-analyticity}

The proof of section \ref{sec:future-directed}, which showed that local modular flow must be future directed, only applied to states $|\Psi\rangle$ satisfying definition \ref{def:weakly-analytic}.
For these ``weakly analytic'' states, there exist local coordinate systems in which the vector $\phi(x) | \Psi \rangle$ can be analytically continued in time.
This condition was needed because the proof hinged on showing that a certain commutator vanishes; we showed that if modular flow is \textit{not} future-directed, then because field operators commute at spacelike separation, there must also exist a pair of null separated points at which field operators commute.
Without analyticity, it is very hard to make an argument that a function should vanish in one regime based on its vanishing somewhere else.

Using a different proof technique, I think it will be possible to prove that local modular flow is future-directed without making such a stringent assumption on the state.
There is a deep relationship between analytic continuation and microlocal smoothness, which is weaker than microlocal analyticity --- see e.g. \cite[theorem 8.1.6]{hormander2003analysis}.
Because modular flow can always be analytically continued, it seems likely that non-future-directed modular flow will violate some basic assumption about the microlocal singularities of correlation functions.
In practice, however, I have found it quite difficult to produce a connection between (i) an analytic continuation of a modular-flowed two-point correlator of smeared operators, and (ii) an analytic continuation of the fundamental two-point distribution of fields.
Sussing out this connection would be an excellent subject for future work.

\subsection{Explicit constructions of states with local modular flow}

It is of considerable interest to know, given a vector field whose flow preserves a spacetime region, whether there is a state in that region whose modular flow implements the vector field.
This paper gave a few necessary conditions that must be satisfied by the vector field, but the settings in which a state with the corresponding modular flow can actually be constructed --- see section \ref{sec:sufficiency-construction} --- were quite limited.
This means that the question ``can this vector field be implemented as a modular flow?'' is highly constrained in that it is easy to answer \textit{no} in many settings, but there are very few settings in which one can answer with a definitive \textit{yes}.
It would be good to develop general tools for constructing states with local modular flow, possibly beginning in free field theory.

\acknowledgments{I thank Stefan Hollands for suggesting that any reasonable microlocal spectrum condition would be inconsistent with spacelike modular flow; I also thank him for being the intended recipient of many emails that I never sent because I figured out the answers while I was writing to ask him.
This work benefited from conversations with Elba Alonso-Monsalve, Daniel Harlow, and Hong Liu, and from communications with Bob Wald.
Financial support was provided by the AFOSR under award number FA9550-19-1-0360, by the DOE Early Career Award number DE-SC0021886, and by the Heising-Simons Foundation.}

\appendix

\section{Contours for extension to an axis}
\label{app:axis-extension}

In section \ref{sec:vacuum-proof}, we encountered a function of two complex variables $s$ and $t$ with the following properties.
\begin{itemize}
	\item The function is defined and holomorphic in the domain $-\frac{1}{2} \leq \text{Im}(s) \leq 0, \text{Im}(t) \geq 0$; in this domain we called the function $F(s, t).$
	\item The function is defined and holomorphic in the domain $0 \leq \text{Im}(s) \leq \frac{1}{2}, \text{Im}(t) \leq 0$; in this domain we called the function $G(s, t).$
	\item The function is continuous in the real $(s, t)$ plane through a domain like the one sketched in figure \ref{fig:discontinuity-domain}, which is an open set contained in the first, second and fourth quadrants and with a boundary curve passing through the origin in the second and fourth quadrants.
\end{itemize}
We claimed that any such function must have a larger domain of continuity through the real $(s, t)$ plane than was originally assumed.
In particular, $F$ and $G$ must meet continuously at any point that can be reached by a curve whose endpoints are in the original domain of continuity, and whose tangent vectors are everywhere in the second or fourth quadrants.
This enlarged domain was sketched in figure \ref{fig:discontinuity-domain-extended}.
We indicated the idea of the proof in section \ref{sec:vacuum-proof}; here we will provide the details.

The global structure of the domains in which $F$ and $G$ are defined is actually not important for us; it does not matter that the imaginary part of $t$ can be taken all the way to $\pm \infty$, or that the imaginary part of $s$ can be taken all the way to $\pm \frac{1}{2}.$
For the sake of generality, we will just assume that $F$ is defined in a small ball in the imaginary second quadrant, and $G$ is defined in a small ball in the imaginary fourth quadrant.
I.e., we will assume that there is a radius $r$ such that $F$ is defined and holomorphic in the domain
\begin{equation}
	B_{F, r}
		= \{ (s, t) | |s|^2 + |t|^2 < r^2\} \cap \{\text{Im}(s) < 0, \text{Im}(t) > 0\},
\end{equation}
and $G$ is defined and holomorphic in the domain
\begin{equation}
	B_{G, r}
	= \{ (s, t) | |s|^2 + |t|^2 < r^2\} \cap \{\text{Im}(s) > 0, \text{Im}(t) < 0\}.
\end{equation}
We will further assume that $F$ and $G$ meet continuously in a small tube around a curve with endpoints in the second and fourth quadrants, and with tangent vectors in the second and fourth quadrants.\footnote{Note that the ``tube'' here plays exactly the same mathematical role as the tube in Borchers's timelike tube theorem \cite{Borchers:tube, Strohmaier:analytic-states, Strohmaier:timelike-tube}; the mathematical argument for the timelike tube theorem is the same as the argument given here for continuity of $F$ and $G$ in an extended domain.}
See figure \ref{fig:appendix-tube-domain}.
Note that all of the assumptions made here are satisfied by the function from section \ref{sec:vacuum-proof}, c.f. figure \ref{fig:discontinuity-domain}.
The assumptions made in this appendix are less restrictive because we only require $F$ and $G$ to be defined for a small open set of values for the real and imaginary parts of $s$ and $t$; this is the setting that is relevant in section \ref{sec:analytic-proof}.

\begin{figure}
	\centering
	\includegraphics{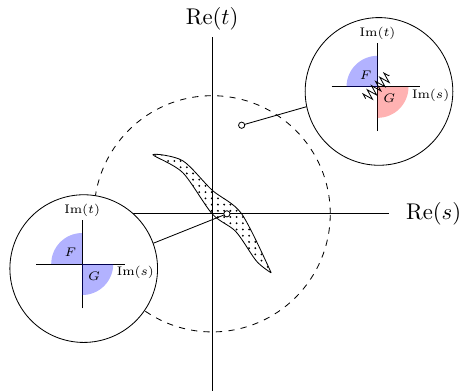}
	\caption{The dashed circle is the circle of radius $r$ in the $(\text{Re}(s), \text{Re}(t))$ plane.
	At each point inside this circle, there is a small open set in the imaginary second quadrant in which the function $F$ is defined, and a small open set in the imaginary fourth quadrant in which the function $G$ is defined.
	These functions are analytic in their domains of definition, and meet continuously at the dotted, ``tube-like'' domain.}
	\label{fig:appendix-tube-domain}
\end{figure}

The ``extended domain'' in which we wish to show that $F$ and $G$ meet continuously consists of all points that can be reached by a curve with endpoints in the tube of continuity and with tangent vectors in the second or fourth quadrants, and such that these points also lie within the ball of radius $r$ around the origin.
We will focus on the portion of this domain lying ``below'' the tube; see figure \ref{fig:appendix-tube-domain-extended}.
To show continuity through the extended domain, we will construct a complex-analytic function $H$ in an open set containing the extended domain, such that $H$ is continuous across the extended domain by construction, and such that $H$ agrees with $F$ and $G$ in certain open subsets of their original domains of definition.
By uniqueness of analytic functions, $F$ and $G$ must be equal to $H$ throughout the extended domain of continuity, and hence equal to each other.

\begin{figure}
	\centering
	\includegraphics{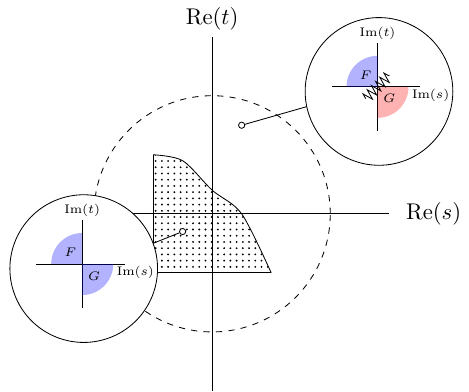}
	\caption{The extended domain of continuity for the functions $F$ and $G$ first shown in figure \ref{fig:appendix-tube-domain}.
	If the circle of radius $r$ were smaller, and the tube were differently shaped, the full ``square-shaped'' dotted domain shown here would need to be truncated at the radius $r.$}
	\label{fig:appendix-tube-domain-extended}
\end{figure}

Let $(s_0, t_0)$ be a point in the extended domain shown in figure \ref{fig:appendix-tube-domain-extended}.
Let $s_{-}$ be the most negative value of $\text{Re}(s)$ achieved in the closure of the tube, and let $s_{+}$ be the most positive value of $\text{Re}(s)$ attained in the tube, so that we have $s_- < s_0 < s_+$.
Similarly, let $t_{-}$ and $t_{+}$ be the appropriate boundaries of the tube in $\text{Re}(t),$ so that we have $t_- < t_0 < t_+.$
See figure \ref{fig:appendix-tube-curve}.
Consider the one-complex-dimensional curve
\begin{align}
	t(s) = \frac{(s_0 - s_{-}) (t_0 - t_-)}{s - s_{-}} + t_-.
\end{align}
A simple calculation shows that the imaginary part of $t$ is given by
\begin{equation}
		\text{Im}(t)
			= - \frac{(s_0 - s_-) (t_0-t_-)}{|s - s_-|^2} \text{Im}(s).
\end{equation}
From this we see that the imaginary part of $t$ is positive for $\text{Im}(s)$ negative, and is negative for $\text{Im}(s)$ positive.

\begin{figure}
	\centering
	\begin{subfigure}[b]{0.48\textwidth}
		\centering
		\includegraphics[width=\textwidth]{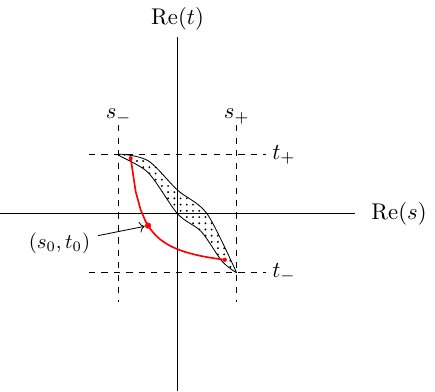}
		\caption{}
		\label{fig:appendix-tube-curve}
	\end{subfigure}
	\hfill
	\begin{subfigure}[b]{0.48\textwidth}
		\centering
		\includegraphics{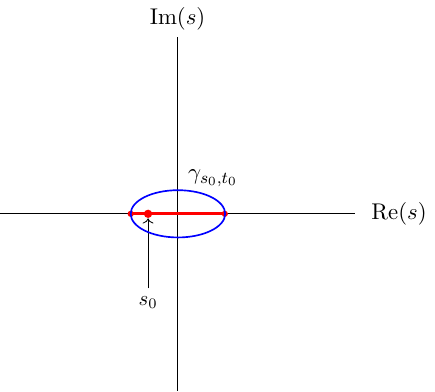}
		\caption{}
		\label{fig:appendix-tube-contour}
	\end{subfigure}
	\caption{(a) The values $s_-, s_+, t_-, t_+$ bounding the tube of continuity in a rectangle, together with a portion of the curve $t = \frac{(t_0 - t_-)(s_0 - s_-)}{s - s_-} + t_-$.
	(b) A contour $\gamma_{s_0, t_0}$ in the complex $s$-plane, containing $s_0,$ that meets the real axis at points where the curve on the left intersects the tube of continuity.}
	\label{fig:tube-curve-and-contour}
\end{figure}

For $s$ real, $t$ goes to infinity as $s$ approaches $s_-,$ so there is some real value of $s$ between $s_-$ and $s_0$ for which $t(s)$ is in the tube of continuity.
There is also clearly some real value of $s$ between $s_0$ and $s_+$ for which $t(s)$ is in the tube of continuity.
See figure \ref{fig:appendix-tube-curve}.
Let $\gamma_{s_0, t_0}$ be a contour in the complex $s$-plane which intersects the real axis at these two points, and which surrounds the point $s_0.$
See figure \ref{fig:appendix-tube-contour}.
For $s$ real and in between these two points of intersection, we clearly have $|t(s)|^2 + |s|^2 < r^2.$
Because $t(s)$ is a continuous function of $s,$ we can always squeeze $\gamma_{s_0, t_0}$ in the imaginary direction to ensure that for every complex number $s$ on the contour $\gamma_{s_0, t_0}$, the bound $|t(s)|^2 + |s|^2 < r^2$ continues to hold.
For any such contour, the full set $(s, t(s))$ for $s \in \gamma_{s_0, t_0}$ is contained in one of the domains $B_{F, r}$ or $B_{G, r}$, except for the two points where $\gamma_{s_0, t_0}$ intersects the real axis, at which $(s, t(s))$ is in the tube of continuity where $F$ and $G$ meet.

Thus far, for every point $(s_0, t_0)$ in the extended domain of continuity, we have demonstrated the existence of a complex disk $t_{s_0, t_0}(s)$ and a contour in the $s$-plane surrounding $s_0$ such that for any $s$ on the contour, the point $(s, t_{s_0, t_0}(s))$ is in the original domain in which $F$ and $G$ were defined.
Because all functions involved depend continuously on $s_0$ and $t_0,$ such a contour will clearly also exist for $(s_0, t_0)$ in a small complex neighborhood of the extended domain of continuity.
Using the disk method of section \ref{subsec:disk-method}, we define a function $H$ in this small complex neighborhood by
\begin{align}
	\begin{split}
	H(s_0, t_0)
		& \equiv \frac{1}{2\pi i} \int_{\text{part of $\gamma_{s_0, t_0}$ lying in $B_{F, r}$}} ds\, \frac{F(s, t_{s_0, t_0}(s))}{s - s_0} \\
		& \qquad + \frac{1}{2\pi i} \int_{\text{part of $\gamma_{s_0, t_0}$ lying in $B_{G, r}$}} ds\, \frac{G(s, t_{s_0, t_0}(s))}{s - s_0}.
	\end{split}
\end{align}
Because $F$ is holomorphic in $B_{F, r},$ and $G$ is holomorphic in $B_{G, r},$ the value of $H(s_0, t_0)$ is unchanged by a small deformation of the contour $\gamma_{s_0, t_0}.$
This means that for any fixed $(s_0, t_0),$ we can pick a small neighborhood of that point in which fixed contour can be used, independent of any specific point in that neighborhood.
The resulting expression for $H(s_0, t_0)$ depends analytically on the points in that neighborhood, and taking complex derivatives shows that $H(s_0, t_0)$ is an analytic function.
All that remains is to show that $H$ agrees with $F$ on a portion of its domain, and with $G$ on a portion of its domain.

To show this, we note that for $(s_0, t_0)$ in the tube of continuity, the full interior of the contour $\gamma_{s_0, t_0}$ is contained in $B_{F, r}, B_{G, r}$ or the tube of continuity.
See figure \ref{fig:appendix-contour-interior}.
The value of $F$ at $(s_0, t_0)$ can be obtained as a limit of points in the portion of the contour whose image lies in $B_{F, r},$ and for any such point $\tilde{s}_0,$ we have (see figure \ref{fig:appendix-interior-perturbation})
\begin{align}
	\begin{split}
		F(\tilde{s}_0, t_{s_0, t_0}(\tilde{s}_0))
		& \equiv \frac{1}{2\pi i} \int_{\text{part of $\gamma_{s_0, t_0}$ lying in $B_{F, r}$}} ds\, \frac{F(s, t_{s_0, t_0}(s))}{s - \tilde{s}_0} \\
		& \qquad + \frac{1}{2\pi i} \int_{\text{real segment closing off the contour}} ds\, \frac{F(s, t_{s_0, t_0}(s))}{s - \tilde{s}_0},
	\end{split}
\end{align}
together with
\begin{align}
	\begin{split}
		0
		& \equiv \frac{1}{2\pi i} \int_{\text{part of $\gamma_{s_0, t_0}$ lying in $B_{G, r}$}} ds\, \frac{G(s, t_{s_0, t_0}(s))}{s - \tilde{s}_0} \\
		& \qquad + \frac{1}{2\pi i} \int_{\text{real segment closing off the contour}} ds\, \frac{G(s, t_{s_0, t_0}(s))}{s - \tilde{s}_0}.
	\end{split}
\end{align}
Adding these together gives $F(\tilde{s}_0, t_{s_0, t_0}(\tilde{s}_0)) = H(\tilde{s}_0, t_{s_0, t_0}(\tilde{s}_0))$, and taking limits gives $F(s_0, t_0) = H(s_0, t_0).$
An identical argument gives $G(s_0, t_0) = H(s_0, t_0).$
We therefore have that $G$ and $F$ agree with $H$ on a portion of their domains, which consequently implies that they agree everywhere; since $H$ was continuous through the extended real domain by construction, $F$ and $G$ must agree continuously on that domain.
This is what we wanted to prove.

\begin{figure}
	\centering
	\begin{subfigure}[b]{0.48\textwidth}
		\centering
		\includegraphics[width=\textwidth]{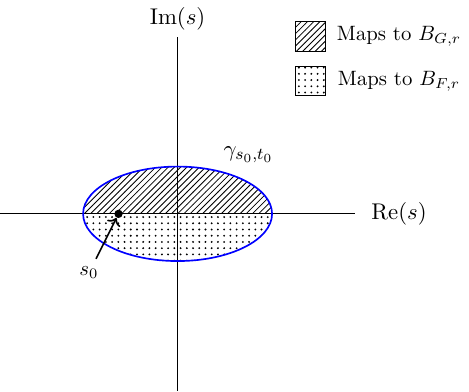}
		\caption{}
		\label{fig:appendix-contour-interior}
	\end{subfigure}
	\hfill
	\begin{subfigure}[b]{0.48\textwidth}
		\centering
		\includegraphics{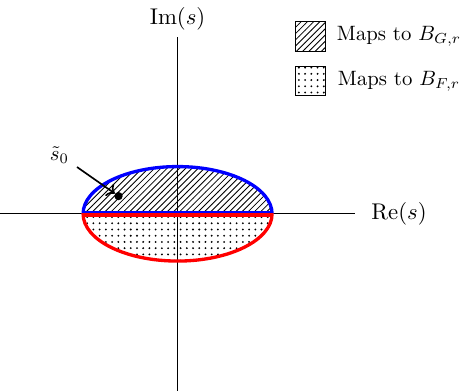}
		\caption{}
		\label{fig:appendix-interior-perturbation}
	\end{subfigure}
	\caption{(a) When $(s_0, t_0)$ is in the domain of continuity, every point in the interior of $\gamma_{s_0, t_0}$ maps into the domain of $F$, the domain of $G$, or the tube where the two functions meet continuously.
	(b) When $\tilde{s}_0$ is moved above the real axis, it lies in the portion of the contour's interior that maps into $B_{F, r}.$
	The contour integral for $H$ can then be split up into a piece surrounding $\tilde{s}_0$ that evaluates to $F$, and a piece that integrates to zero.}
	\label{fig:edge-of-the-wedge}
\end{figure}

\bibliographystyle{JHEP}
\bibliography{bibliography}

\end{document}